\documentclass[english,twocolumn,transaction]{IEEEtran}
\usepackage[T1]{fontenc}
\usepackage{float}
\usepackage{amsmath}
\usepackage{amsthm}
\usepackage{amssymb}
\usepackage{stackrel}
\usepackage{graphicx}
\usepackage{setspace}
\usepackage{color}
\usepackage{url}

\makeatletter

\providecommand{\tabularnewline}{\\}
\floatstyle{ruled}
\newfloat{algorithm}{tbp}{loa}
\providecommand{\algorithmname}{Algorithm}
\floatname{algorithm}{\protect\algorithmname}

\theoremstyle{plain}

\theoremstyle{definition}

\theoremstyle{plain}

\theoremstyle{plain}

\IEEEoverridecommandlockouts
\usepackage{ifpdf}
\usepackage{cite}
\ifCLASSOPTIONcompsoc
\usepackage[caption=false,font=normalsize,labelfont=sf,textfont=sf]{subfig}
\else
\usepackage[caption=false,font=footnotesize]{subfig}
\fi
\usepackage{amsmath}

\@ifundefined{showcaptionsetup}{}{
\PassOptionsToPackage{caption=false}{subfig}}
\usepackage{subfig}
\makeatother

\usepackage{babel}

\newtheorem{theo}{Theorem}
\newtheorem{remark}{Remark}
\renewcommand\figurename{Fig.}
\newcommand*{\QEDA}{\hfill\ensuremath{\blacksquare}}
\begin{document}
\title{Analysis and Optimization of Service Delay for Multi-quality Videos in Multi-tier Heterogeneous Network with Random Caching}

\author{Xuewei Zhang, Tiejun Lv, ~\IEEEmembership{Senior Member,~IEEE}, Yuan Ren,\\
Wei Ni, ~\IEEEmembership{Senior Member,~IEEE}, and Norman C. Beaulieu, ~\IEEEmembership{Fellow,~IEEE}
\thanks{
The financial support of the National Natural Science Foundation of
China (NSFC) (Grant No. 61671072 and 61801382) and the Beijing Natural Science
Foundation (No. L192025) is gratefully acknowledged.
(\emph{Corresponding author: Tiejun Lv.})

X. Zhang and Y. Ren are with the School of Communication and Information
Engineering, Xi'an University of Posts and Telecommunications (XUPT), Xi'an
710121, China (e-mail: zhangxw@bupt.edu.cn, renyuan@xupt.edu.cn).
X. Zhang is also with the School of Information and Communication
Engineering, Beijing University of Posts and Telecommunications (BUPT), Beijing
100876, China.

T. Lv is with the School of Information and Communication
Engineering, Beijing University of Posts and Telecommunications (BUPT), Beijing
100876, China (e-mail: lvtiejun@bupt.edu.cn).

W. Ni is with Data61, Commonwealth Scientific and Industrial Research,
Sydney, NSW 2122, Australia (e-mail: wei.ni@data61.csiro.au).

N. C. Beaulieu is with the School of Information and Communication Engineering
and the Beijing Key Laboratory of Network System Architecture and
Convergence, Beijing University of Posts and Telecommunications, Beijing
100876, China (e-mail: nborm@bupt.edu.cn).
}}

\maketitle

\begin{abstract}
Aiming to minimize service delay,
we propose a new random caching scheme in device-to-device (D2D)-assisted heterogeneous network.
To support diversified viewing qualities of multimedia video services,
each video file is encoded into a base layer (BL) and multiple enhancement layers (ELs) by scalable video coding (SVC).
A super layer, including the BL and several ELs, is transmitted to every user.
We define and quantify the service delay of multi-quality videos by deriving successful transmission probabilities
when a user is served by a D2D helper, a small-cell base station (SBS) and a macro-cell base station (MBS).
We formulate a delay minimization problem subject to the limited cache sizes of D2D helpers and SBSs.
The structure of the optimal solutions to the problem is revealed,
and then an improved standard gradient projection method is designed to effectively obtain the solutions.
Both theoretical analysis and Monte-Carlo simulations validate
the successful transmission probabilities.
Compared with three benchmark caching policies,
the proposed SVC-based random caching scheme is superior in terms of reducing the service delay.
\end{abstract}
\section{Introduction}
\renewcommand\figurename{Fig.}
The service modes of wireless communications are transferring from connection-oriented services \cite{Gupta2012H},
such as voice call and short message, to content-oriented services,
such as on-demand multimedia video \cite{Zhang2018Energy2}.
The amount of data traffic is experiencing a more significant surge than ever before.
It is predicted that the total
amount of data traffic will reach 100 exabytes by 2023,
and multimedia video services will account for most of the 100 exabytes \cite{2016Ericsson}.
Under this circumstance,
backhaul with finite bandwidth is expected to become increasingly restrictive
when retrieving requested contents from the core network to wireless edges \cite{Zhang2018Near},
i.e., co-existing base stations (BSs) and user equipments.
The limited capacity of backhaul is one of the most restrictive factors,
especially for time-sensitive and real-time video services.
Aiming to relieve this pressing limitation of backhaul
and mitigate service delays,
wireless caching is proposed as a promising technique,
and attracts strong attention in the Fifth Generation (5G) communication networks and beyond
\cite{Tao2015Content}.

With proactive caching enabled in wireless edges,
video files requested by users can be pre-fetched to the local storage of wireless edges via backhauls \cite{Chen2017Cooperative}.
The content placement is performed in light-traffic time periods.
Cached contents can be delivered to the users, if requested.
According to the types of cached contents,
wireless caching can be typically classified into
uncoded caching and coded caching.
Earlier studies focused on the design of uncoded caching \cite{Zhang2017Multicast,Tao2015Content},
in which uncoded video files, especially those popular ones, are placed in the local caches of wireless edges.
Later, wireless caching is extended to coded caching \cite{Xu2017Fundamental,Lampiris2018Adding},
where complete videos are firstly encoded into different data packets
and then these coded packets are locally stored by the proposed caching strategies.

Among many caching schemes, random caching, also known as probabilistic caching,
is an important class of wireless caching \cite{Wen2017Random, cui2016analysis,Zheng2017Probabilistic,Zhang2018Energy},
where complete video files or their combinations are prefetched to be cached under a certain caching distribution
which can be optimized.
In \cite{Wen2017Random} and \cite{cui2016analysis},
by optimizing the successful transmission probabilities,
the random caching distributions were determined.
The authors of \cite{Zheng2017Probabilistic} derived the content hit probability and its approximation for throughput analysis.
By maximizing these two metrics, the caching probabilities were optimized.
In our recent paper \cite{Zhang2018Energy},
we studied random caching in heterogeneous network.
The random caching probabilities were optimized to maximize the energy efficiency of the considered network.

With no assistance of BSs, device-to-device (D2D) communications allow users to establish direct links with their nearby neighbors.
This helps reduce the overall transmission power of the system,
and improve the system throughput \cite{Wu2017Energy, Ahmed2018A}.
By integrating wireless caching into D2D communications,
data traffic can be offloaded from small-cell BSs (SBSs) and macro-cell BSs (MBSs),
relieving traffic congestion and reducing service delay \cite{chen2016cache,Chen2018Caching,Deng2018The,Zhang2016Efficient}.
Chen \emph{et al}. \cite{chen2016cache} evaluated the offloading gain and energy cost of D2D helpers,
when the offloading opportunity was maximized.
In \cite{Chen2018Caching}, a machine learning model was proposed to capture the content popularity and request preference
in D2D communications.
The authors of \cite{Deng2018The} focused on the energy cost of D2D helpers,
and proposed two hybrid caching schemes to reduce the cost.
To optimize the system throughput,
Zhang \emph{et al}. \cite{Zhang2016Efficient} took
both D2D-link scheduling and resource allocation into account in single-hop D2D communications.

Given the limited backhaul capacity,
ever-changing channel conditions and varying user requirements,
multi-quality video services are in increasingly high demand,
including multimedia services for standard definition videos (SDVs) and high definition videos (HDVs).
To provide diversified perceptual viewing experiences to mobile users,
scalable video coding (SVC), developed for advanced video coding (AVC) \cite{Schwarz2007Overview},
has attracted a lot of interest.
With the aid of SVC, each video can be divided into
a base layer (BL) and several enhancement layers (ELs) \cite{Guo2018Multi}.
The BL contains the most basic information of the scalable video,
and the file only containing the BL can be decoded as SDV, which has the lowest viewing quality.
Successive ELs, together with the BL, can provide HDV.
More layers provide better video quality,
and the video with all divided layers can exhibit the most excellent viewing quality \cite{Ostovari2015Scalable}.
More technical details for the encoding and decoding process of SVC can refer to \cite{Schwarz2007Overview}.
SVC has been applied to wireless caching in the literature.
The authors of \cite{Zhan2018SVC} maximized the total throughput of cache-enabled heterogeneous network
by jointly optimizing SVC-based retrieving decision and data rate allocation.
In \cite{Zhang2017Layered}, given the layered structure of video files,
the data traffic delivered over backhaul was minimized.
In our earlier work \cite{Zhang2018Near},
we proposed an SVC-based layer placement scheme
and maximized the average amount of offloaded traffic,
so that most data traffic was retrieved from SBSs and the pressure was relieved on the MBS.

For large-scale video transmissions, the limited backhaul capacity is often the bottleneck of the system.
Congestions in backhaul would lead to unacceptable latency.
Hence, effective performance metric of service delay is crucial,
and needs to be carefully designed \cite{Chen2015Delay}.
Relying on queuing theory,
the authors of \cite{Amer2018Inter} derived the average delay for unit request,
and minimized the delay with the greedy algorithm.
A weighted average delay for unit request was considered in \cite{Amer2018Optimizing},
through which the bandwidth allocation and caching probability distribution were yielded.
In \cite{Li2018Learning}, a learning-based caching scheme was devised in D2D-assisted network,
with the objective of minimizing the average transmission delay.
The delay was also minimized by jointly designing the caching and user association strategies in \cite{Wang2016Joint}.
As mentioned earlier, mobile users can request different viewing qualities according to their preference or network states,
while the study of SVC-supported wireless caching is still in a very earlier stage.
On the other hand, provided SVC is in place,
the unnecessary video layers may not need to be delivered.
This can significantly reduce the service delay.
Therefore, delay analysis of SVC-based video retrievals is important.

This paper presents a new random caching scheme in D2D-assisted three-tier heterogeneous network,
consisting of D2D, SBS and MBS tiers, to minimize the service delay.
To provide diversified viewing qualities of video services,
each video file is encoded by SVC.
A super layer, containing the BL and several ELs, is delivered for providing multi-quality multimedia video services.
A user can be served by the nearest D2D helper or SBS which caches the requested super layers.
When requested contents cannot be locally provided,
the nearest MBS is responsible for retrieving the contents
from the core network via its backhaul at the additional cost of resource and latency.

In the proposed SVC-based random caching scheme,
D2D helpers and SBSs randomly select super layers to cache,
and the caching probabilities can adapt to the delay performance of the three-tier heterogeneous network.
The key contributions can be summarized as follows.
\begin{itemize}
\item Any requested videos are encoded by SVC,
and super layers, each of which consists of a BL and several successive ELs, are cached randomly.
By sending super layers to mobile users, diversified viewing qualities can be achieved.
This can avoid SVC decoding at the users, hence reducing the service delay,
as compared to conventional separate transmissions of different layers.
\item We define the partial service delays
when a user is served by a D2D helper, an SBS, or an MBS.
We derive the corresponding successful transmission probabilities,
from which the expressions for partial service delays can be attained,
and so can the overall averaged service delay.
\item Subject to the limited cache sizes of D2D helpers and SBSs,
the delay minimization problem is formulated.
The structure of the optimal solution to the delay minimization problem is discussed.
An improved standard gradient projection method is accordingly developed to provide sub-optimal solutions for the random caching probabilities.
\end{itemize}

The rest of this paper is arranged as follows.
Section II presents the network model, channel model and SVC-based random caching scheme.
In Section III, we first define the service delay,
and then derive the successful transmission probabilities to establish the expression for service delay.
In Section IV, the delay minimization problem is formulated and solved.
Numerical results are presented in Section V.
Finally, concluding remarks are provided in Section VI.
\section{System Model}
In this section, we provide the system model of the considered heterogeneous network,
including network model, channel model and SVC-based random caching protocol.
\subsection{Network Model}
The three-tier heterogeneous
network we consider consists of an MBS tier, an SBS tier and a D2D tier.
The SBSs and D2D helpers have limited cache sizes,
and can store parts of requested video contents during off-peak hours.
The MBSs, SBSs and D2D helpers are single-antenna and  randomly distributed,
following independently and identically distributed (i.i.d.) Poisson Point Processes (PPPs)
$\mathrm{\Phi_{m}}$, $\mathrm{\Phi_{s}}$ and $\mathrm{\Phi_{d}}$
with density parameters $\lambda_{\rm{m}}$, $\lambda_{\rm{s}}$ and $\lambda_{\rm{d}}$, respectively.
Some of the users act as D2D helpers;
and the rest can  receive the requested contents from the D2D helpers, SBSs or MBSs.
Without loss of generality,
we randomly choose one user that demands a video service as the typical user,
and set the position of this user as the origin of the observed network
\cite{Wen2017Random,Chen2017Cooperative}\footnote{
The mobility model is not explicitly considered in this paper.
Nevertheless, the statistical analysis presented in this paper can
implicitly capture the mobility of users.
This is because, under a given PPP realization,
any random position shift of this PPP can be regarded as another new PPP realization.
Considering the user-centric scenario, the movement of
the user can be interpreted to the position shift of the BSs with reference to this user.
}.
Taking the location of the typical user as the center,
the D2D helpers are distributed in the circular area with radius $r_{\rm{c}}$ \cite{chen2016cache,Chen2018Caching}\footnote{
The D2D helpers often own limited communication capacities,
since the transmit power of these helpers is kept in a low level compared to SBSs and MBSs.
Owing to this fact, the potential serving helpers are confined in the area that are close to the user.}.
Compared with ordinary D2D equipments,
the adopted D2D helper is endowed with limited storage capacity.
When the requested video contents can be found in the local caches of D2D helper,
it can provide cached contents to its neighboring users by D2D communications.
With the aid of cache-enabled D2D helpers,
more multimedia data traffic can be locally provided,
thus alleviating the heavy traffic burden of BSs and backhaul links.
The SBSs that are likely to serve the user are distributed in the cell with radius $r_{\rm{d}}$.
We assume $r_{\rm{d}}\geq r_{\rm{c}}$.
Likewise, limited cache storage is also allocated to each SBS.
The D2D helpers operate in the overlay mode \cite{Zhang2016Socially}.
The spectrums allocated to the BSs and D2D helpers are orthogonal.
When the typical user is receiving contents from the D2D helper,
interference only comes from other non-serving helpers.
The MBSs and SBSs are also assigned with orthogonal spectrums.
Therefore, the inter-tier interference can be avoided \cite{Zhang2017Cost}.
The D2D helpers are labeled in the ascending order of their distances from the user,
and so are SBSs and MBSs.

Line-of-sight (LoS) propagation typically dominates
small-scale wireless communications \cite{Wildemeersch2014Successive,Xu2017Modeling}.
Therefore, we assume that the nearest of the nodes storing the requested content is chosen as the serving one.
Specifically, in the collaborative area,
the nearest D2D helper which caches the requested video content, denoted by $\rm{d}_{0}$,
is selected as the serving helper.
The detailed content request and user connection steps are illustrated as follows.
At the beginning of each content transmission process,
the MBS is able to collect user requests,
and then the MBS broadcasts these video requests to D2D helpers.
Afterwards, each helper checks its content library to determine whether the requested contents are locally cached.
If some helpers store the requested content,
they will reply to the MBS, and the MBS will choose the nearest helper in user's collaborative area as the serving one.
Based on the steps shown above, the communication link between this serving helper and the content-demanding user is established,
and then the requested contents can be delivered to the user via D2D communications\footnote{
In the D2D-assisted cache-enabled heterogeneous networks,
the content-demanding users and D2D helpers can connect to the MBS for synchronization.
When the synchronization process finishes, the connections between them will be suspended,
and then the D2D data transmission process can be performed.
If the control/user plane decoupled network is considered, the content-demanding users and D2D helpers
are able to connect to the MBS in the control plane for synchronization,
and then D2D data communication can be completed in the user plane, as proposed in \cite{Song2017Control}.}.
If all D2D helpers in the collaborative area fail to provide the requested video file,
the nearest of the SBSs storing the content, denoted by $\rm{s}_{0}$, serves the user.
In the worst case where the requested video cannot be retrieved from the local caches of D2D helpers and SBSs,
the nearest MBS, denoted by $\rm{m}_{0}$, retrieves the content for the user
via its capacity-constrained backhaul.
For illustration convenience,
the D2D helpers and SBSs that cache the requested video files are referred to as \emph{potential serving} D2D helpers and SBSs, respectively.
An example of the considered network model is provided in Fig. 1.
\begin{figure}[t]
\centering{}\includegraphics[scale=0.25]{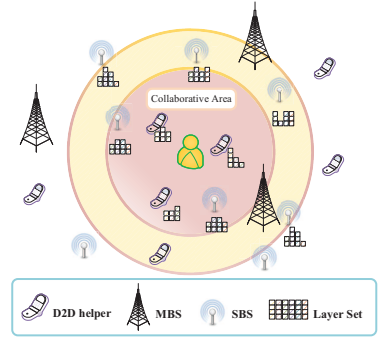}
\caption{In the considered network,
the typical user is located at the center,
and can be served by a D2D helper, an SBS or an MBS.
For example, in this model, each video file is divided into four layers.}
\label{System_Model}
\end{figure}
\subsection{Channel Model}
In modern wireless cellular networks, the users and BSs are densely deployed,
and the co-channel interference generally dominates over the background additive white Gaussian noise at the user's side  \cite{Xu2017Modeling,Wen2017Random}.
Thus, we consider the interference-limited network,
where the noise can be negligible compared with the severe co-channel interference.
When the typical user can obtain the requested video file from the nearest potential serving D2D helper,
the received signal-to-interference ratio (SIR) at the user can be written as
\begin{align}
\mathrm{SIR}_{\rm{d}}=\frac{\left|h_{0}^{\rm{d}}\right|^{2}(r_{0}^{\rm{d}})^{-\alpha_{\rm{d}}}}{\sum_{k\in\mathrm{\Phi_{d}}\setminus \rm{d}_{0}}\left|h_{k}^{\rm{d}}\right|^{2}(r_{k}^{\rm{d}})^{-\alpha_{\rm{d}}}}
\end{align}
where $h_{0}^{\rm{d}}$ and $h_{k}^{\rm{d}}$ are the small-scale fading
from the D2D helper $\rm{d}_{0}$ and other non-serving D2D helpers,
obeying the complex Gaussian distribution with zero mean and unit variance,
i.e., $\mathcal{CN}(0,1)$;
$r_{0}^{\rm{d}}$ is the distance between $\rm{d}_{0}$ and the user;
$r_{k}^{\rm{d}}$ is the distance between helper $k$ and the user;
and $\alpha_{\rm{d}}$ is the path loss exponent from the D2D helpers to the user, which satisfies $\alpha_{\rm{d}}>2$.

If the user fails to retrieve the contents from the local cache of the D2D helpers,
the nearest potential serving SBS is assigned to transmit the videos.
The received SIR achieved by SBS $\rm{s}_{0}$ can be written as
\begin{align}
\mathrm{SIR}_{\rm{s}}=\frac{\left|h_{0}^{\rm{s}}\right|^{2}(r_{0}^{\rm{s}})^{-\alpha_{\rm{s}}}}{\sum_{k\in\mathrm{\Phi_{s}}\setminus \rm{s}_{0}}\left|h_{k}^{\rm{s}}\right|^{2}(r_{k}^{\rm{s}})^{-\alpha_{\rm{s}}}}
\end{align}
where $h_{0}^{\rm{s}}$ and $h_{k}^{\rm{s}}$ are the small-scale fading
from the SBS $\rm{s}_{0}$ and other non-serving SBSs, following $\mathcal{CN}(0,1)$;
$r_{0}^{\rm{s}}$ is the distance between $\rm{s}_{0}$ and the user;
$r_{k}^{\rm{s}}$ is the distance between SBS $k$ and the user;
and $\alpha_{\rm{s}}$ is the path loss exponent from the SBSs to the user, satisfying $\alpha_{\rm{s}}>2$.

When the requested video file needs to be retrieved from the core network through backhaul,
the user is served by the nearest MBS $\rm{m}_{0}$.
The SIR delivered by $\rm{m}_{0}$ can be written as
\begin{align}
\mathrm{SIR}_{\rm{m}}=\frac{\left|h_{0}^{\rm{m}}\right|^{2}(r_{0}^{\rm{m}})^{-\alpha_{\rm{m}}}}{\sum_{k\in\mathrm{\Phi_{m}\setminus}\rm{m}_{0}}\left|h_{k}^{\rm{m}}\right|^{2}(r_{k}^{\rm{m}})^{-\alpha_{\rm{m}}}}
\end{align}
where $h_{0}^{\rm{m}}$ and $h_{k}^{\rm{m}}$ are the small-scale fading from the nearest MBS $\rm{m}_{0}$ and other farther MBSs,
following $\mathcal{CN}(0,1)$; $r_{0}^{\rm{m}}$ is the distance between MBS $\rm{m}_{0}$ and the user;
$r_{k}^{\rm{m}}$ is the distance between MBS $k$ and the user;
and $\alpha_{\rm{m}}>2$ is the path loss exponent from the MBSs to the user.
\subsection{SVC-based Caching Scheme }
We aim to provide multi-quality multimedia video services to mobile users.
The requested video files are pre-processed by SVC into multiple layers,
and the number of divided layers is $L$.
Among these layers, $l=1$ indicates the BL,
and $l=2,...,L$ indicate the ELs.
The $L$-level video qualities can be provided to the user.
The video with quality level $l$ consists of
the content layers indexed from $1$ to $l$, a BL and $(l-1)$ successive ELs.
We define the video with quality level $l$ as super layer $l$.
The user with the super layer can decode the requested video with the preferred viewing quality.
In our recent works \cite{Zhang2018Near,Zhang2019Economical},
the BSs were designed to cooperatively transmit different video layers to the user.
At the user, successive interference cancelation (SIC) was employed to separate the individual signals from multiple serving SBSs,
and then the requested video file was decoded.
This would cause extra decoding latency and resource overhead.
In this research, with the super layer,
the user does not need to run SIC and SVC decoding, reducing computational complexity and service delay.

A video library is located in the core network, where there are $F$ video files
which can be requested by the user.
By applying SVC, each video file is divided into $L$ layers.
The size of each layer is $s_{f,l}$.
The size of super layer $l$ from the $f$-th video is $c_{f,l}=\sum_{k=1}^{l}s_{f,k}$\footnote{
In this paper, we consider the more general case,
where the sizes of video files or content layers do not need to be the same.
This means that the file size or the layer size has no effect on the design of caching protocol,
performance analysis and problem formulation.}.
The videos are ranked in the descending popularity order.
The Zipf law has been widely used to characterize the request probability \cite{breslau1999web};
however, this distribution is inadequate for wireless video requests.
To fill this void, the Mandelbrot-Zipf (M-Zipf) law has been developed for cellular video requests \cite{Lee2019Throughput},
where the request probability of the $f$-th video is given by
\begin{gather}
p(f)=\frac{(f+q)^{-\alpha}}{\sum_{n=1}^{F}(n+q)^{-\alpha}},\ f=1,2,...,F,\label{eq:c-1-1}
\end{gather}
\noindent where $\alpha$ is the skewness parameter to account for the degree of request concentration \cite{breslau1999web};
and $q$ is the plateau factor.
With larger $q$, there would be a smaller difference among the request probabilities of the most popular files.
If $q=0$, the M-Zipf distribution becomes the Zipf distribution.
The values of both $\alpha$ and $q$ can be experimentally determined through real-world data sets,
which is beyond the scope of this paper.

Regarding the request probability of video quality,
the preference for SDV of the $f$-th video is $g_{\rm{SDV}}(f)=\frac{f-1}{F-1}$ \cite{Wu2016Caching},
and therefore the preference for HDV is $g_{\rm{HDV}}(f)=1-g_{\rm{SDV}}(f)$.
When HDV is requested, all quality levels are assumed to have the same request popularity.
To this end, the probability of the perceptual preference for super layer $l$ from the $f$-th video is given by
\begin{gather}
p_{f,l}=\begin{cases}
p(f)\cdotp\frac{f-1}{F-1}, & l=1;\\
p(f)\cdotp\frac{F-f}{(F-1)(L-1)}, & l=2,...,L.
\end{cases}
\end{gather}

In the considered network,
D2D helpers and SBSs are capable of caching and transmitting super layers to the typical user,
where the caching probabilities are unknown.
In this paper, we design the random caching scheme for super layers in heterogeneous network.
The $f$-th video file with quality level $l$ is randomly placed in the local caches of
D2D helpers and SBSs with probabilities $p_{f,l}^{\rm{d}}$ and $p_{f,l}^{\rm{s}}$, respectively.
The optimized random caching probabilities are collected by the MBS,
and the MBS can determine whether to cache the super layers in D2D helpers and SBSs or not.
The super layers can be locally placed according to these probabilities.
With larger caching probability, the super layer is more likely to be cached and vise verse.
The probabilities will be optimized later.
For notational convenience, all values of $p_{f,l}^{\rm{d}}$ and $p_{f,l}^{\rm{s}}$ are collected in the matrices
$\bf{p^{\rm{d}}}\in\mathbb{R^{\mathrm{F\times L}}}$ and $\bf{p^{\rm{s}}}\in\mathbb{R^{\mathrm{F\times L}}}$, respectively.
\section{Service Delay Analysis}
In this section, we first establish the expressions for partial service delays
when the typical user is served by the nearest potential serving D2D helper, SBS or MBS.
To achieve the expressions for these delays,
the successful transmission probabilities are derived,
and in turn the overall service delay is obtained.
\subsection{Performance Analysis of Service Delay}
When referring to service delay,
the transmission of the requested super layer is successful.
The service delay is based on the successful transmission.
A successful transmission refers to the case that the received SIR
at the user exceeds a pre-defined quality-of-service (QoS) requirement,
so that the minimum data rate for data transmission can be guaranteed.

The D2D helpers cache the super layer $l$ of the $f$-th video file with probability $p_{f,l}^{\rm{d}}$,
and the potential serving D2D helpers caching the layer form a thinning PPP with density $p_{f,l}^{\rm{d}}\lambda_{\rm{d}}$.
According to the property of PPP, when there is no serving D2D helper in the collaborative area,
the probability is $\exp\left(-\lambda_{\rm{d}}p_{f,l}^{\rm{d}}\pi r_{\rm{c}}^{2}\right)$.
On the contrary, when there is at least one serving D2D helper in the area,
this probability is
\begin{align}
&a_{f,l}^{\rm{d}}=1-\exp\left(-\lambda_{\rm{d}}p_{f,l}^{\rm{d}}\pi r_{\rm{c}}^{2}\right)\label{pp1}
\end{align}
which is also called the D2D association probability.
Based on this, when the user is served by the nearest potential serving D2D helper,
the partial service delay is given by
\begin{align}
&D_{f,l}^{\rm{d}}=a_{f,l}^{\rm{d}}P_{f,l}\left(\mathrm{SIR}_{\rm{d}}\geq\theta\right)\frac{c_{f,l}}{W_{\rm{d}}\log_{2}(1+\theta)}\label{pp1}
\end{align}
\noindent where $P_{f,l}\left(\mathrm{SIR}_{\rm{d}}\geq\theta\right)$ is the successful transmission probability
when the serving D2D helper delivers super layer $l$ of the $f$-th video;
$\theta$ is the minimum QoS requirement of the user;
and $W_{\rm{d}}$ is the allocated system bandwidth for D2D tier.

Likewise, when there is no serving SBS in the cell,
the probability is $\exp\left(-\lambda_{\rm{s}}p_{f,l}^{\rm{s}}\pi r_{\rm{d}}^{2}\right)$,
and the SBS association probability is given by
\begin{align}
&a_{f,l}^{\rm{s}}=1-\exp\left(-\lambda_{\rm{s}}p_{f,l}^{\rm{s}}\pi r_{\rm{d}}^{2}\right).
\end{align}
There are two cases that the user will be served by the SBS.
The first case is that the user cannot find any D2D helper in the collaborative area with probability $1-a_{f,l}^{\rm{d}}$.
The second case is that the user can connect to the D2D helper
but the received signal strength cannot meet the minimum QoS requirement,
and the probability of this case is $a_{f,l}^{\rm{d}}(1-P_{f,l}\left(\mathrm{SIR}_{\rm{d}}\geq\theta\right))$.
Thus, the transmission delay when the user is served by the nearest potential serving SBS can be expressed as
\begin{align}
D_{f,l}^{\rm{s}}=&\left[(1-a_{f,l}^{\rm{d}})+a_{f,l}^{\rm{d}}(1-P_{f,l}\left(\mathrm{SIR}_{\rm{d}}\geq\theta\right))\right]\nonumber\\
&a_{f,l}^{\rm{s}}P_{f,l}\left(\mathrm{SIR}_{\rm{s}}\geq\theta\right)\frac{c_{f,l}}{W_{\rm{s}}\log_{2}(1+\theta)}\nonumber\\
=&\left(1-a_{f,l}^{\rm{d}}P_{f,l}\left(\mathrm{SIR}_{\rm{d}}\geq\theta\right)\right)\nonumber\\
&a_{f,l}^{\rm{s}}P_{f,l}\left(\mathrm{SIR}_{\rm{s}}\geq\theta\right)\frac{c_{f,l}}{W_{\rm{s}}\log_{2}(1+\theta)}\label{pp2}
\end{align}
\noindent where $P_{f,l}\left(\mathrm{SIR}_{\rm{s}}\geq\theta\right)$ is the successful transmission probability
when the serving SBS transmits the requested super layer to the user;
and $W_{\rm{s}}$ is the allocated system bandwidth for SBS tier.

When neither the potential serving D2D helpers nor SBSs can provide the requested super layer,
the nearest MBS becomes the serving node.
The requested super layer is first retrieved via backhaul from the core network,
and then sent to the user.
The service delay is caused by backhaul retrieval and downlink transmission.
To this end, the partial service delay is written as
\begin{align}
D_{f,l}^{\rm{m}}\nonumber=&\Big[1-a_{f,l}^{\rm{d}}P_{f,l}\left(\mathrm{SIR}_{\rm{d}}\geq\theta\right)-a_{f,l}^{\rm{s}}P_{f,l}\left(\mathrm{SIR}_{\rm{s}}\geq\theta\right)\nonumber\\
&+a_{f,l}^{\rm{d}}P_{f,l}\left(\mathrm{SIR}_{\rm{d}}\geq\theta\right)a_{f,l}^{\rm{s}}P_{f,l}\left(\mathrm{SIR}_{\rm{s}}\geq\theta\right)\Big]\nonumber\\
&\left(\frac{c_{f,l}}{R_{\rm{bh}}}+P\left(\mathrm{SIR}_{\rm{m}}\geq\theta\right)\frac{c_{f,l}}{W_{\rm{m}}\log_{2}(1+\theta)}\right)\nonumber\\
=&\left(1-a_{f,l}^{\rm{d}}P_{f,l}\left(\mathrm{SIR}_{\rm{d}}\geq\theta\right)\right)\left(1-a_{f,l}^{\rm{s}}P_{f,l}\left(\mathrm{SIR}_{\rm{s}}\geq\theta\right)\right)\nonumber\\
&c_{f,l}\left(\frac{1}{R_{\rm{bh}}}+\frac{P\left(\mathrm{SIR}_{\rm{m}}\geq\theta\right)}{W_{\rm{m}}\log_{2}(1+\theta)}\right)\label{pp3}
\end{align}
\noindent where $P\left(\mathrm{SIR}_{\rm{m}}\geq\theta\right)$ is the successful transmission probability from the nearest MBS;
$W_{\rm{m}}$ is the system bandwidth assigned for MBSs;
and $R_{\rm{bh}}$ is the backhaul data rate.

By taking the three types of service delay into account,
the overall service delay for a successful transmission event is written as
\begin{align}
 D\left(\bf{p^{\rm{d}}},\bf{p^{\rm{s}}}\right)&=\sum_{f=1}^{F}\sum_{l=1}^{L}p_{f,l}\left(D_{f,l}^{\rm{d}}+D_{f,l}^{\rm{s}}+D_{f,l}^{\rm{m}}\right).\label{eq:Delay}
\end{align}
From (\ref{pp1}) to (\ref{eq:Delay}),
we have the following findings:
\begin{itemize}
\item The random caching probabilities of the D2D helpers and SBSs
affect the association probabilities and the successful transmission probabilities,
and thus have a strong impact on the service delay.
In the following, given the limited cache sizes of D2D helpers and SBSs,
the random caching vectors ${\bf{p^{\rm{d}}}}$ and ${\bf{p^{\rm{s}}}}$ are meticulously designed.
\item The three partial service delays are not parallel with each other, and there are some trade-offs among them.
When more required contents can be found in local caches,
the user can be locally served by the D2D helper or the SBS.
In this case, the MBS is less likely to serve the user, and the time-consuming backhaul delivery can be avoided.
Thus, the overall service delay can be largely reduced.
\end{itemize}
\subsection{Successful Transmission Probabilities}
In Section III-A, we present the delay expressions for three cases,
in each of which the successful transmission probability is expected to be derived.

When the user is served by the nearest potential serving D2D helper,
two cases are summarized in the following:
\begin{itemize}
\item Case 1: In the collaborative area,
the geographically nearest D2D helper is the serving D2D helper,
and the probability of having the helper is $p_{f,l}^{\rm{d}}$.
\item Case 2: The geographically nearest helper is not the serving D2D helper,
and the serving helper is one of the farther helpers in the observed collaborative area.
The probability for this case is $\left(1-p_{f,l}^{\rm{d}}\right)$.
\end{itemize}

Bearing the two cases in mind,
the successful transmission probability of the D2D helper is written as
\begin{align}
&P_{f,l}\left(\mathrm{SIR}_{\rm{d}}\geq\theta\right)\nonumber \\
&=p_{f,l}^{\rm{d}}P_{f,l}^{1}\left(\mathrm{SIR}_{\rm{d}}\geq\theta\right)+\left(1-p_{f,l}^{\rm{d}}\right)P_{f,l}^{2}\left(\mathrm{SIR}_{\rm{d}}\geq\theta\right)
\end{align}
where $P_{f,l}^{1}\left(\mathrm{SIR}_{\rm{d}}\geq\theta\right)$ and $P_{f,l}^{2}\left(\mathrm{SIR}_{\rm{d}}\geq\theta\right)$
are the successful transmission probabilities with respect to Case 1 and Case 2, respectively.
Their expressions are provided in Theorem 1.

\begin{theo}
In the collaborative area with radius $r_{\rm{c}}$,
when the typical user is served by the nearest potential serving D2D helper,
the successful transmission probabilities for Case 1 and Case 2 are give respectively by
\begin{align}
P_{f,l}^{1}\left(\mathrm{SIR}_{\rm{d}}\geq\theta\right)=\begin{cases}
\frac{p_{f,l}^{\rm{d}}q_{f,l}^{\rm{d}}\big(\theta^{-\frac{2}{\alpha_{\rm{d}}}}\big)}{1-\exp\left(-\lambda_{\rm{d}}p_{f,l}^{\rm{d}}\pi r_{\rm{c}}^{2}\right)},  & if \; 0<p_{f,l}^{\rm{d}}\leq1;\\
0,  & if \; p_{f,l}^{\rm{d}}=0,
\end{cases}
\end{align}
\begin{align}
P_{f,l}^{2}\left(\mathrm{SIR}_{\rm{d}}\geq\theta\right)=\begin{cases}
\frac{p_{f,l}^{\rm{d}}q_{f,l}^{\rm{d}}(0)}{1-\exp\left(-\lambda_{\rm{d}}p_{f,l}^{\rm{d}}\pi r_{\rm{c}}^{2}\right)}, & if \; 0<p_{f,l}^{\rm{d}}\leq1;\\
0, & if \; p_{f,l}^{\rm{d}}=0,
\end{cases}
\end{align}
\noindent where
\begin{align}
&q_{f,l}^{\rm{d}}(x)=\frac{1-\exp\left(-\lambda_{\rm{d}}\pi r_{\rm{c}}^{2}\left(p_{f,l}^{\rm{d}}+\theta^{\frac{2}{\alpha_{\rm{d}}}}G_{\alpha_{\rm{d}}}(x)\right)\right)}{p_{f,l}^{\rm{d}}+\theta^{\frac{2}{\alpha_{\rm{d}}}}G_{\alpha_{\rm{d}}}(x)},\\
&G_{a}(b)=\int_{b}^{\infty}\frac{1}{1+x^{\frac{a}{2}}}\mathrm{d}x.
\end{align}
\end{theo}
\begin{IEEEproof}
Refer to Appendix A.
\end{IEEEproof}

\begin{remark}
The successful transmission probability is higher in Case 1 than that in Case 2.
The reason is provided as follows.
In Case 1, the interference only comes from the D2D helpers that are farther away than $\rm{d}_{0}$;
however, in Case 2, the interference includes what is generated by the helpers closer than the serving node $\rm{d}_0$.
Therefore, the interference from Case 2 is severer,
leading to a lower successful transmission probability.
The same finding can be obtained
when the user is served by the nearest potential serving SBS.
\end{remark}

According to Theorem 1,
the expression for $P_{f,l}\left(\mathrm{SIR}_{\rm{d}}\geq\theta\right)$ can be written as
\begin{align}
P_{f,l}\left(\mathrm{SIR}_{\rm{d}}\geq\theta\right)&=\frac{p_{f,l}^{\rm{d}}}{1-\exp\left(-\lambda_{\rm{d}}p_{f,l}^{\rm{d}}\pi r_{\rm{c}}^{2}\right)}\nonumber\\
&\left(p_{f,l}^{\rm{d}}\left(q_{f,l}^{\rm{d}}\left(\theta^{-\frac{2}{\alpha_{\rm{d}}}}\right)-q_{f,l}^{\rm{d}}(0)\right)+q_{f,l}^{\rm{d}}(0)\right).\label{STP_111}
\end{align}

Similarly, there are two cases in regards of the successful transmission
when the typical user is served by the nearest potential serving SBS.
The two cases are shown as follows:
\begin{itemize}
\item Case 3: The geographically nearest SBS caches the requested video layer,
and becomes the serving SBS. The probability is $p_{f,l}^{\rm{s}}$.
\item Case 4: The geographically nearest SBS does not cache the requested super layer.
The serving node is one of the farther SBSs,
with the probability $\left(1-p_{f,l}^{\rm{s}}\right)$.
\end{itemize}

By taking the two cases into consideration, $P_{f,l}\left(\mathrm{SIR}_{\rm{s}}\geq\theta\right)$ can be written as
\begin{align}
&P_{f,l}\left(\mathrm{SIR}_{\rm{s}}\geq\theta\right)\nonumber\\
&=p_{f,l}^{\rm{s}}P_{f,l}^{3}\left(\mathrm{SIR}_{\rm{s}}\geq\theta\right)+\left(1-p_{f,l}^{\rm{s}}\right)P_{f,l}^{4}\left(\mathrm{SIR}_{\rm{s}}\geq\theta\right)
\end{align}
\noindent where $P_{f,l}^{3}\left(\mathrm{SIR}_{\rm{s}}\geq\theta\right)$ and $P_{f,l}^{4}\left(\mathrm{SIR}_{\rm{s}}\geq\theta\right)$
are the successful transmission probabilities in Case 3 and Case 4, respectively.
The expressions for the probabilities can be found in the following theorem.
\begin{theo}
If the typical user is served by the nearest potential serving SBS,
the successful transmission probabilities of Case 3 and Case 4
are given by
\begin{align}
P_{f,l}^{3}\left(\mathrm{SIR}_{\rm{s}}\geq\theta\right)=\begin{cases}
\frac{p_{f,l}^{\rm{s}}q_{f,l}^{s}\big(\theta^{-\frac{2}{\alpha_{\rm{s}}}}\big)}{1-\exp\left(-\lambda_{\rm{s}}p_{f,l}^{\rm{s}}\pi r_{\rm{d}}^{2}\right)}, & if \; \text{0<}p_{f,l}^{\rm{s}}\leq1;\\
0, & if \; p_{f,l}^{\rm{s}}=0,
\end{cases}
\end{align}
\begin{align}
P_{f,l}^{4}\left(\mathrm{SIR}_{\rm{s}}\geq\theta\right)=\begin{cases}
\frac{p_{f,l}^{\rm{s}}q_{f,l}^{\rm{s}}(0)}{1-\exp\left(-\lambda_{\rm{s}}p_{f,l}^{\rm{s}}\pi r_{\rm{d}}^{2}\right)}, & if \; \text{0<}p_{f,l}^{\rm{s}}\leq1;\\
0, & if \; p_{f,l}^{\rm{s}}=0,
\end{cases}
\end{align}
\noindent where
\begin{align}
&q_{f,l}^{\rm{s}}(x)=\frac{1-\exp\left(-\lambda_{\rm{s}}\pi r_{\rm{d}}^{2}\left(p_{f,l}^{\rm{s}}+\theta^{\frac{2}{\alpha_{\rm{s}}}}G_{\alpha_{\rm{s}}}(x)\right)\right)}{p_{f,l}^{\rm{s}}+\theta^{\frac{2}{\alpha_{\rm{s}}}}G_{\alpha_{\rm{s}}}(x)}.
\end{align}
\end{theo}
\begin{IEEEproof}
The proof can be obtained by referring to Appendix A. We omit detailed steps due to limited space.
\end{IEEEproof}
Based on Theorem 2, the expression for $P_{f,l}(\mathrm{SIR}_{\rm{s}}\geq\theta)$ can be written as
\begin{align}
P_{f,l}\left(\mathrm{SIR}_{\rm{s}}\geq\theta\right)&=\frac{p_{f,l}^{\rm{s}}}{1-\exp\left(-\lambda_{\rm{s}}p_{f,l}^{\rm{s}}\pi r_{\rm{d}}^{2}\right)}\nonumber\\
&\left(p_{f,l}^{\rm{s}}\left(q_{f,l}^{\rm{s}}\left(\theta^{-\frac{2}{\alpha_{\rm{s}}}}\right)-q_{f,l}^{\rm{s}}(0)\right)+q_{f,l}^{\rm{s}}(0)\right).\label{STP_222}
\end{align}

When the nearest MBS is assigned to deliver the requested super layers,
the successful transmission probability is shown in the following theorem.
\begin{theo}
The successful transmission probability for the nearest MBS serving the typical user is given by
\begin{gather}
P\left(\mathrm{SIR}_{\rm{m}}\geq\theta\right)=\left[1+\theta^{\frac{2}{\alpha_{\rm{m}}}}G_{\alpha_{\rm{m}}}\left(\theta^{-\frac{2}{\alpha_{\rm{m}}}}\right)\right]^{-1}.\label{STP_333}
\end{gather}
\end{theo}
\begin{IEEEproof}
See Appendix B.
\end{IEEEproof}

\begin{remark}
Notice that, when $a=4$, we can obtain that $G_{a}(b)=\pi/2-\arctan(b)={\rm{arccot}}(b).$
Thus, if $\alpha_{\rm{d}}=\alpha_{\rm{s}}=\alpha_{\rm{m}}=4$,
the closed-form expressions for (\ref{eq:Delay}), (\ref{STP_111}), (\ref{STP_222}) and (\ref{STP_333}) can be attained.
Otherwise, they cannot be regarded as closed-form ones.
\end{remark}

\begin{remark}
From (\ref{STP_111}) and (\ref{STP_222}), we observe that the PPP densities $\lambda_{\rm{d}}$ and $\lambda_{\rm{s}}$
have a strong impact on the
successful transmission probabilities $P_{f,l}\left(\mathrm{SIR}_{\rm{d}}\geq\theta\right)$ and $P_{f,l}\left(\mathrm{SIR}_{\rm{s}}\geq\theta\right)$, respectively.
From (\ref{STP_333}), it is seen that $P\left(\mathrm{SIR}_{\rm{m}}\geq\theta\right)$
is independent of the PPP density $\lambda_{\rm{m}}$.
We can infer that,
the PPP density can affect the average number of serving nodes, and, in turn, the successful transmission probabilities $P_{f,l}\left(\mathrm{SIR}_{\rm{d}}\geq\theta\right)$ and $P_{f,l}\left(\mathrm{SIR}_{\rm{s}}\geq\theta\right)$.
The probabilities also depend on the path-loss exponent and caching probability.
\end{remark}

As a result, the expressions for successful transmission probabilities are obtained
in (\ref{STP_111}), (\ref{STP_222}), and (\ref{STP_333}), and so are the partial service delays.
By substituting (\ref{pp1})-(\ref{pp3}) into (\ref{eq:Delay}),
the expression for overall service delay is developed.
\section{The Problem Formulation and the Proposed Algorithm}
In this section, we minimize the service delay
subject to the finite cache sizes of D2D helpers and SBSs.
Then, we exploit the structure of the optimal solutions,
and develop the improved standard gradient projection method,
from which the sub-optimal caching probabilities can be achieved.
\subsection{Problem Formulation}
Given the expression for the overall service delay (\ref{eq:Delay}),
the delay minimization problem is formulated as
\begin{subequations}\label{max_1}
\begin{align}
\underset{\bf{p^{\rm{d}}}, \bf{p^{\rm{s}}}}{\mathrm{min}}\quad\: &  D\left(\bf{p^{\rm{d}}},\bf{p^{\rm{s}}}\right)\label{eq:objective_1}\\
\mathrm{\mathrm{s.t.}}\:\quad\: & \sum_{f=1}^{F}\sum_{l=1}^{L}p_{f,l}^{\rm{d}}c_{f,l}\leq M_{\rm{d}},\label{eq:cons_1}\\
&\sum_{f=1}^{F}\sum_{l=1}^{L}p_{f,l}^{\rm{s}}c_{f,l}\leq M_{\rm{s}},\label{eq:cons_2}
\end{align}
\begin{align}
 & 0\leq p_{f,l}^{\rm{d}}\leq1, \forall f,\forall l, \label{eq:cons_3}\\
 & 0\leq p_{f,l}^{\rm{s}}\leq1, \forall f,\forall l.\label{eq:cons_4}
\end{align}
\end{subequations}
\noindent Constraints (\ref{eq:cons_1}) and (\ref{eq:cons_2})
indicate the cache size restrictions,
where $M_{\rm{d}}$ and $M_{\rm{s}}$ are the cache sizes allocated to each D2D helper and SBS, respectively.
Furthermore, inequalities (\ref{eq:cons_3}) and (\ref{eq:cons_4}) specify the feasible solution regions
of the caching probabilities $p_{f,l}^{\rm{d}}$ and $p_{f,l}^{\rm{s}}$, respectively.
\subsection{Proposed Algorithm}
To solve the problem (\ref{max_1}),
we show the structure of the optimal solutions in the following theorem.
\begin{theo}
The random caching probabilities $p_{f,l}^{\rm{d}}$ and $p_{f,l}^{\rm{s}}$ are maximized,
if and only if the cache sizes of each D2D helper and SBS are fully utilized;
or in other words, constraints (\ref{eq:cons_1}) and (\ref{eq:cons_2}) take equality.
\end{theo}

\begin{proof}
The overall service delay includes three parts, i.e., the partial delays caused by D2D helpers, SBSs and backhaul deliveries.
In practice, the backhaul retrieval is much more time-consuming,
since there are limited backhaul resources shared by many users.
In other words, $D_{f,l}^{\rm{m}}\gg D_{f,l}^{\rm{s}}$ and $D_{f,l}^{\rm{m}}\gg D_{f,l}^{\rm{d}}$.
Note that the values of $D_{f,l}^{\rm{s}}$ and $D_{f,l}^{\rm{d}}$ are comparable.
To reduce the service delay, it is an effective manner to improve the content hit rate \cite{Wu2018Energy},
which is the probability that requested contents can be satisfied by local caches.
With the improved content hit rate, the user is more likely to acquire the requested super layers locally.
Thus, the backhaul deliveries can be avoided and the service delay can be shortened.

By using the proposed caching and transmission schemes,
the content hit rate of super layer $l$ from the $f$-th file is given by
\begin{align}
&p^{\rm{hit}}_{f,l}\nonumber\\
&=1-\left(1-a_{f,l}^{\rm{d}}P_{f,l}\left(\mathrm{SIR}_{\rm{d}}\geq\theta\right)\right)\left(1-a_{f,l}^{\rm{s}}P_{f,l}\left(\mathrm{SIR}_{\rm{s}}\geq\theta\right)\right). \nonumber
\end{align}
It can be seen that $p^{\rm{hit}}_{f,l}$ is a monotonically increasing function
of $p_{f,l}^{\rm{d}}$ and $p_{f,l}^{\rm{s}}$.
With larger caching probabilities,
the content hit rate can be improved, and the service delay can be reduced to a lower level.
Therefore, the minimum service delay can be attained when the cache sizes of D2D helpers and SBSs are fully utilized.
\end{proof}

According to Theorem 4,
the original optimization problem (\ref{max_1}) can be reformulated as
\begin{subequations}\label{max_2}
\begin{align}
\underset{\bf{p^{\rm{d}}}, \bf{p^{\rm{s}}}}{\mathrm{min}}\quad\: &  D\left(\bf{p^{\rm{d}}},\bf{p^{\rm{s}}}\right)\label{eq:objective_2}\\
\mathrm{\mathrm{s.t.}}\:\quad\: & \sum_{f=1}^{F}\sum_{l=1}^{L}p_{f,l}^{\rm{d}}c_{f,l}=M_{\rm{d}},\label{eq:cons_2_1}\\
&\sum_{f=1}^{F}\sum_{l=1}^{L}p_{f,l}^{\rm{s}}c_{f,l}=M_{\rm{s}},\label{eq:cons_2_2}\\
&(\ref{eq:cons_3}), (\ref{eq:cons_4})\label{eq:cons_2_3}.
\end{align}
\end{subequations}
\begin{algorithm}[t]
\begin{enumerate}
\item Initialization: Set $t=1$, and $\epsilon(1)=1$. Input the accuracy threshold
$\Delta$, and find $p_{f,l}^{\rm{d}}$ and $p_{f,l}^{\rm{s}}$ feasible
to (\ref{eq:cons_2_1})-(\ref{eq:cons_2_3}).
\item Set $\delta\ll\Delta$, and set the maximum number of iterations as $T$.
\item while ($t\leq T$ and $\delta\leq\Delta$)
\item $\quad$For all $f\in\{1,...,F\}$ and $l\in\{1,...,L\}$, calculate
\par\setlength\parindent{1em} $\frac{\partial  D\left(\bf{p^{\rm{d}}},\bf{p^{\rm{s}}}\right)}{\partial p_{f,l}^{\rm{d}}}$, and then obtain
\begin{align}
&\hat{p}_{f,l}^{\rm{d}}(t+1)=p_{f,l}^{\rm{d}}(t)\nonumber\\
&-\epsilon(t)\frac{\partial  D\left(\bf{p^{\rm{d}}},\bf{p^{\rm{s}}}\right)}{\partial p_{f,l}^{\rm{d}}}\mid_{p_{f,l}^{\rm{d}}=p_{f,l}^{\rm{d}}(t),p_{f,l}^{\rm{s}}=p_{f,l}^{\rm{s}}(t)}.\label{eq:Alg_1}
\end{align}
\item $\quad$For all $f\in\{1,...,F\}$ and $l\in\{1,...,L\}$, calculate
\begin{equation}
p_{f,l}^{\rm{d}}(t+1)=\mathrm{min}\left\{\left[\hat{p}_{f,l}^{\rm{d}}(t+1)-u_{1}\right]^{+},1\right\},\label{eq:Alg_2}
\end{equation}
\par\setlength\parindent{1em} where $u_{1}$ satisfies
\begin{equation}
\sum_{_{f=1}}^{F}\sum_{_{l=1}}^{L}c_{f,l}\mathrm{min}\left\{\left[\hat{p}_{f,l}^{\rm{d}}(t+1)-u_{1}\right]^{+},1\right\}=M_{\rm{d}}.\label{eq:Alg_3}
\end{equation}
\item $\quad$For all $f\in\{1,...,F\}$ and $l\in\{1,...,L\}$, evaluate
\par\setlength\parindent{1em} $\frac{\partial D\left(\bf{p^{\rm{d}}},\bf{p^{\rm{s}}}\right)}{\partial p_{f,l}^{\rm{s}}}$,
and then obtain
\begin{align}
&\hat{p}_{f,l}^{\rm{s}}(t+1)=p_{f,l}^{\rm{s}}(t)\nonumber\\
&-\epsilon(t)\frac{\partial  D\left(\bf{p^{\rm{d}}},\bf{p^{\rm{s}}}\right)}{\partial p_{f,l}^{\rm{s}}}\mid_{p_{f,l}^{\rm{d}}=p_{f,l}^{\rm{d}}(t),p_{f,l}^{\rm{s}}=p_{f,l}^{\rm{s}}(t)}.\label{eq:Alg_4}
\end{align}
\item $\quad$ For all $f\in\{1,...,F\}$ and $l\in\{1,...,L\}$, evaluate
\begin{equation}
p_{f,l}^{\rm{s}}(t+1)=\mathrm{min}\{[\hat{p}_{f,l}^{\rm{s}}(t+1)-u_{2}],1\},\label{eq:Alg_5}
\end{equation}
\par\setlength\parindent{1em} where $u_{2}$ satisfies
\begin{equation}
\sum_{_{f=1}}^{F}\sum_{_{l=1}}^{L}c_{f,l}\mathrm{min}\left\{\left[\hat{p}_{f,l}^{\rm{s}}(t+1)-u_{2}\right]^{+},1\right\}=M_{\rm{s}}.\label{eq:Alg_6}
\end{equation}
\item $\quad$$\Delta=\left|D\left({\bf{p^{\rm{d}}}},{\bf{p^{\rm{s}}}}\right)^{(t+1)}- D\left({\bf{p^{\rm{d}}}},{\bf{p^{\rm{s}}}}\right)^{(t)}\right|.$
\item $\quad$$t=t+1$ and $\epsilon(t)=\frac{1}{t}$.
\item end
\end{enumerate}
\caption{The standard gradient projection method for solving Problem (\ref{max_2})}
\end{algorithm}
Note that the objective (\ref{eq:objective_2}) is differentiable;
however, it is difficult to tell (\ref{eq:objective_2}) is convex due to its complex form.
Constraints (\ref{eq:cons_2_1}) to (\ref{eq:cons_2_3}) are linear
equalities and inequalities with regard to $p_{f,l}^{\rm{s}}$ and $p_{f,l}^{\rm{d}}$.
They form a convex variable set.
In light of this, the standard gradient projection method can be employed to solve this problem.
The standard gradient projection method is an effective approach to solve the optimization problem
with a differentiable objective function over a convex variable set \cite{Wen2017Random}.
The detailed procedure of solving this problem is summarized in Algorithm 1.

In Algorithm 1, Steps 4) and 6) calculate the partial derivatives of the service delay
with respect to $p_{f,l}^{\rm{d}}$ and $p_{f,l}^{\rm{s}}$\footnote{
These partial derivatives can be obtained by conventional derivative methods,
whose forms are complicated and therefore omitted in this paper.
When running Algorithm 1, the derivatives can be directly calculated by Matlab with ``diff'' function.}.
Given a pre-defined step size $\epsilon(t)$,
$\hat{p}_{f,l}^{\rm{d}}(t+1)$ and $\hat{p}_{f,l}^{\rm{s}}(t+1)$ are updated according to formulas (\ref{eq:Alg_1}) and (\ref{eq:Alg_4}).
They can be regarded as the optimal values of $p_{f,l}^{\rm{d}}$ and $p_{f,l}^{\rm{s}}$
achieved at the $(t+1)$-th iteration without the cache size restrictions
of each D2D helper and SBS.
In order to not exceed the cache sizes,
the values of $\hat{p}_{f,l}^{\rm{d}}(t+1)$ and $\hat{p}_{f,l}^{\rm{s}}(t+1)$
are projected onto the variable sets comprised of the cache size restrictions (\ref{eq:cons_2_1}) and (\ref{eq:cons_2_2}),
as done in (\ref{eq:Alg_2}) and (\ref{eq:Alg_5}).
$[\:\cdot\:]^{+}$ is employed to guarantee non-negativity.
Thus, $\mathrm{min}\left\{\left[\cdot\right]^{+},1\right\}$ can maintain valid feasible value regions of $p_{f,l}^{\rm{d}}$ and $p_{f,l}^{\rm{s}}$,
satisfying (\ref{eq:cons_2_3}).
As the result of adopting Algorithm 1,
sub-optimal solutions for the random caching probabilities
can be achieved while all constraints in (\ref{max_2}) are satisfied.
\begin{figure*}[t]
\centering{}

\subfloat[$P_{f,l}\left(\mathrm{SIR}_{\rm{d}}\geq\theta\right)$ versus $p_{f,l}^{{\rm{d}}}$]{\includegraphics[scale=0.36]{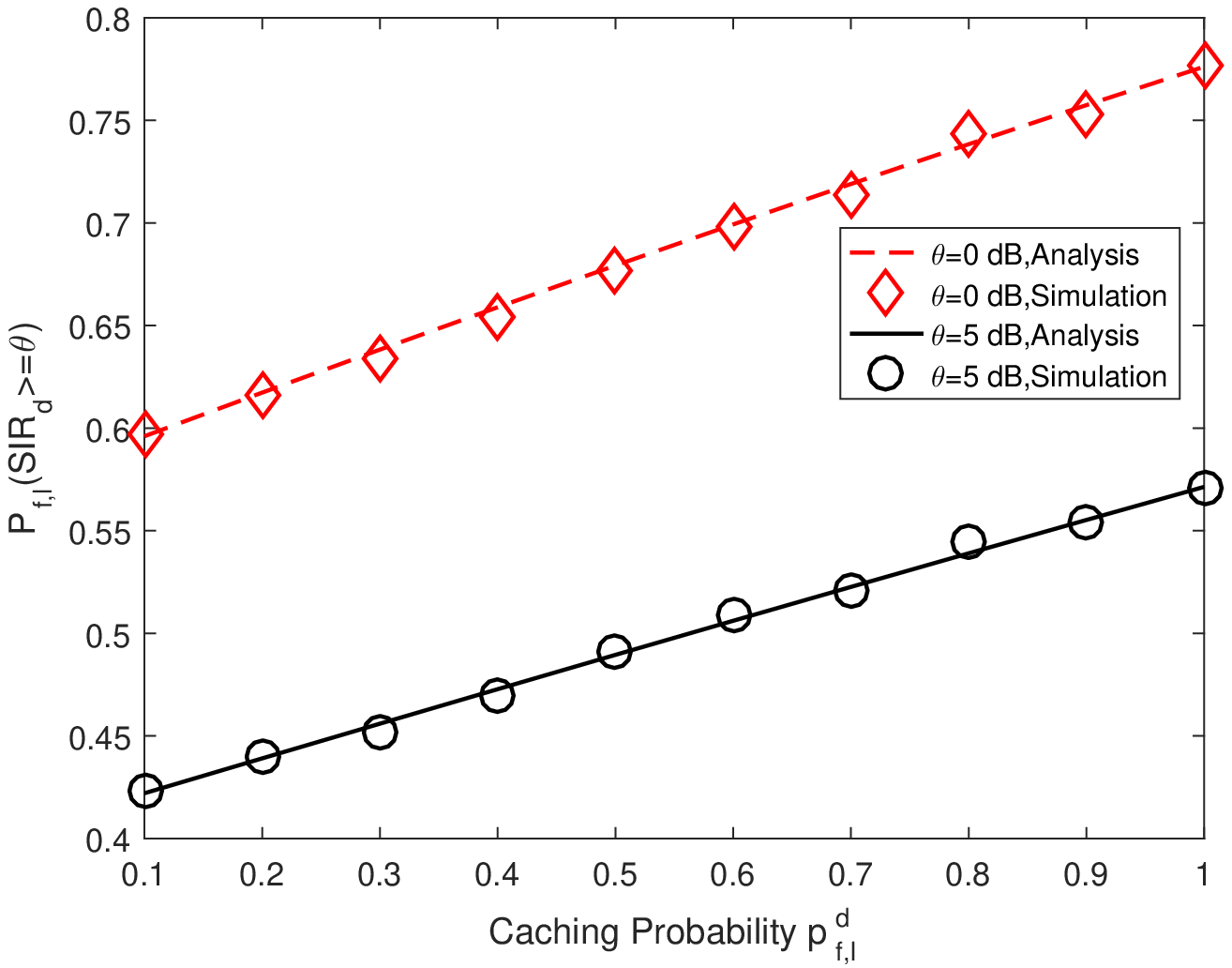}

}\subfloat[$P_{f,l}\left(\mathrm{SIR}_{\rm{s}}\geq\theta\right)$ versus $p_{f,l}^{{\rm{s}}}$]{\includegraphics[scale=0.36]{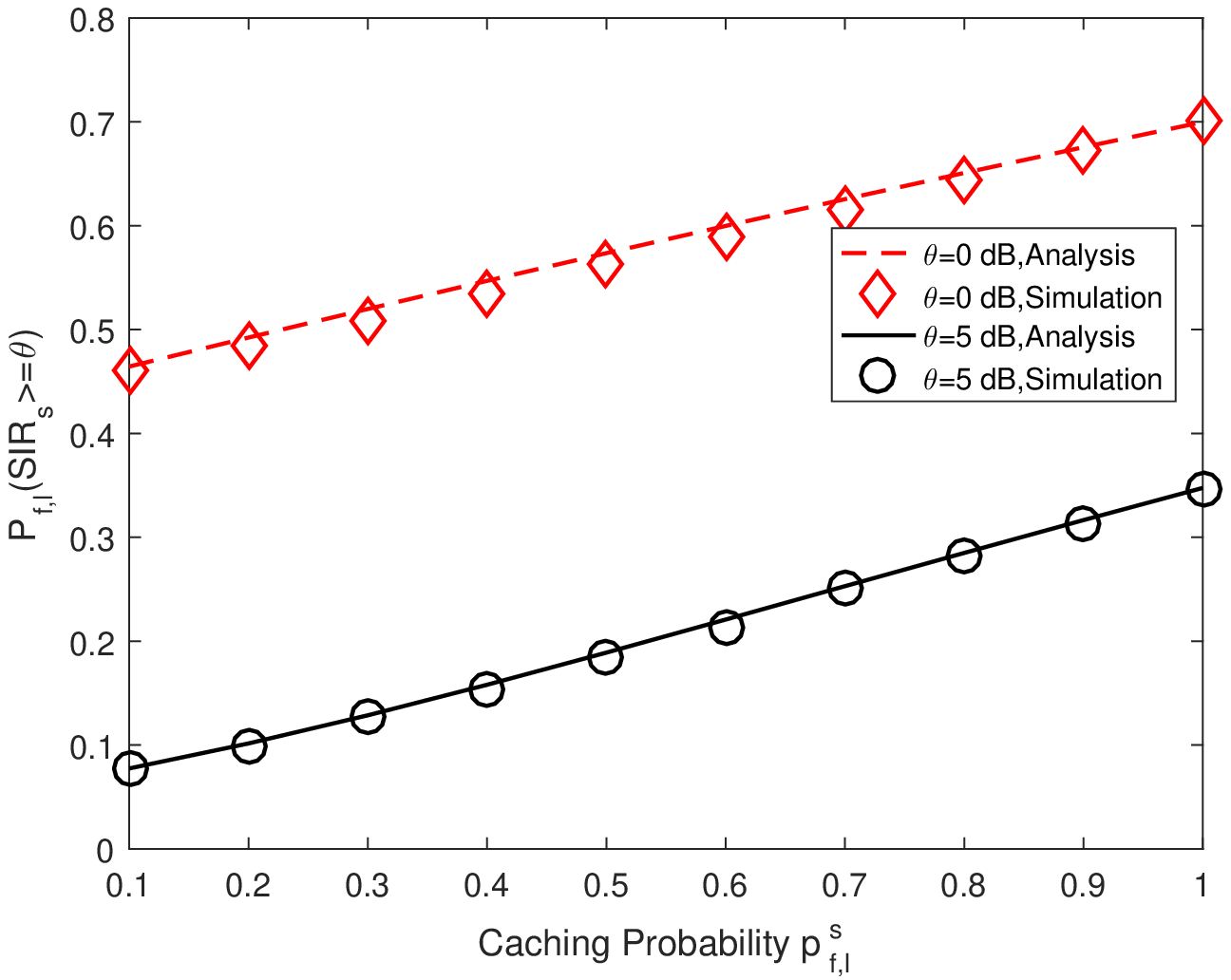}

}\subfloat[$P\left(\mathrm{SIR}_{\rm{m}}\geq\theta\right)$ versus $\theta$]{\includegraphics[scale=0.36]{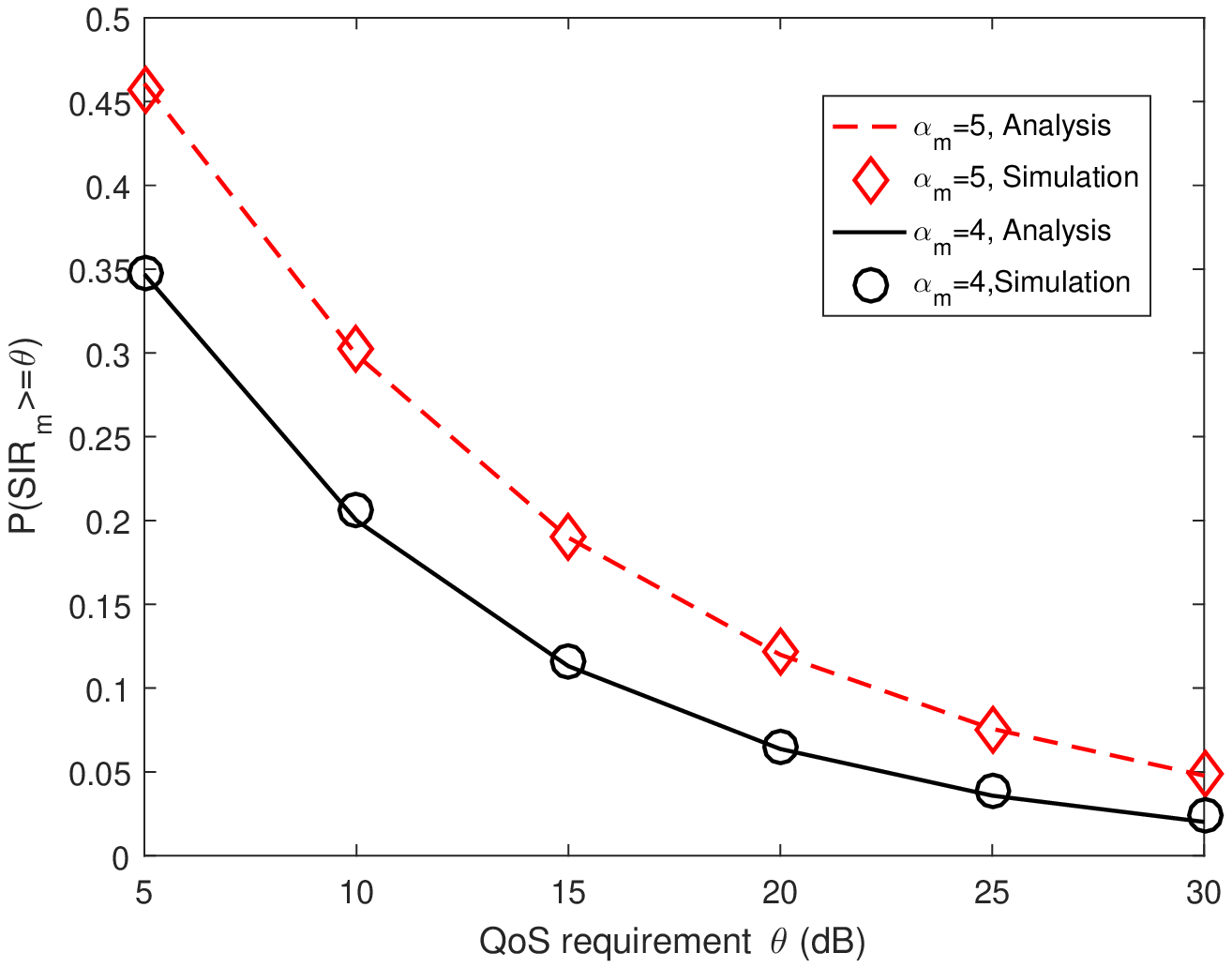}
}
\caption{The successful transmission probabilities when the typical user is served by the D2D helper, SBS and MBS.}
\end{figure*}
We proceed to discuss the property of the proposed algorithm.
In the feasible variable region comprised by constraints (\ref{eq:cons_2_1}) to (\ref{eq:cons_2_3}),
the optimal point is selected at the negative direction with the fastest descent at each iteration.
Thereby, the obtained value of $D\left({\bf{p^{\rm{d}}}},{\bf{p^{\rm{s}}}}\right)$ is
less than or at least equal to the value sourced from the last iteration,
and this algorithm will certainly converge.
From the simulation results shown later,
the proposed algorithm can converge to the sub-optimal solution within a small number of iterations.
At each iteration, the cost for calculating the partial derivatives dominates over the cost of the proposed algorithm,
and the total required number of calculating the partial derivatives scales with $FL$.
When using ${\cal O}$ function for computational analysis, according to \cite{GriewankEvaluating,Zhang2008Computing}.
Thus, when using ${\cal O}$ function for complexity analysis, 
the overall computational complexity of the proposed algorithm is ${\cal O}(FL)$.

\begin{table}[t]
\centering{}
\caption{The Values of Simulation Parameters}
\begin{tabular}{c|c}
\hline
Parameters & Values\tabularnewline
\hline
$\theta$& $5$ dB\cite{Song2017OptimalContent}\tabularnewline
\hline
$F$& $20$ \cite{Muller2016Smart}\tabularnewline
\hline
$L$& $2$\cite{Xiang2018Secure}\tabularnewline
\hline
$s_{f,l}$& $25$ Mbits\cite{Wu2018Energy}\tabularnewline
\hline
$q$ & $5$ \tabularnewline
\hline
$\alpha$ & $1$ \cite{Zhang2020Double}\tabularnewline
\hline
$M_{\rm{d}}$, $M_{\rm{s}}$ & $200$ Mbits $500$ Mbits \tabularnewline
\hline
$\lambda_{\rm{d}}$, $\lambda_{\rm{s}}$ & $0.01$ \cite{chen2016cache}, $0.001$ \tabularnewline
\hline
$\alpha_{\rm{d}}$, $\alpha_{\rm{s}}$, $\alpha_{\rm{m}}$ & 4 \cite{Yang2016Analysis}\tabularnewline
\hline
\end{tabular}
\end{table}
\section{Simulation Results}
In order to show the correctness of the derived successful transmission probabilities,
similar to \cite{Afshang2016Modeling},
the numerical results of both analysis and Monte-Carlo (MC) simulations for these expressions are displayed.
The MC experiments for each point are performed more than 50,000 times.
Next, we show the numerical results of the service delay.
For comparison, three benchmark strategies are adopted, as follows:
\begin{itemize}
\item Most popular content placement (MPCP): The super layers from the most popular videos
are locally cached in the D2D helpers and SBSs.
This caching scheme is used in most uncoded caching,
e.g., in \cite{Zhang2017Multicast,Tao2015Content}.
\item Equal probability content placement (EPCP): In the local caches of D2D helpers and SBSs,
the super layers of all video files are randomly stored with the same caching probabilities
until all cache sizes are fully utilized.
This scheme overlooks the file popularity distribution and perceptual quality preference.
\item Independent content placement (ICP):
The D2D helpers and SBSs randomly select different super layers to cache,
irrespective to the actual service requirements of users.
\end{itemize}
The values of simulation parameters are listed in Table I.
We have given the references for most values of the simulation parameters,
and the other values are set accordingly.
In specific, the density of the SBS is set as 10 times of the D2D helpers;
the range of the plateau factor should satisfy $0\leq q \leq F$, and we set $q=5$;
and the cache sizes of the D2D helpers and SBSs are designed according to the layer size and the number of video files.

In Fig. 2, we plot the successful transmission probabilities.
It is shown that the performance gap is negligible between our analysis and MC simulations.
This validates (\ref{STP_111}), (\ref{STP_222}) and (\ref{STP_333}).
In Figs. 2(a) and 2(b),
it is revealed that larger caching probabilities have stronger positive effects on successful transmission probabilities,
since the closer serving D2D helper and SBS can be found around the user,
providing stronger received signal strength.
From Fig. 2(c), we can conclude that a larger QoS requirement leads to a lower successful transmission probability.
Furthermore, a smaller path loss exponent can adversely affect the successful transmission probability.
This is because, though the received signal strength increases as the path loss exponent decreases,
it fails to scale with the growing interference caused by other MBSs.

\begin{figure}[t]
\begin{center}
\centering{}\includegraphics[scale=0.4]{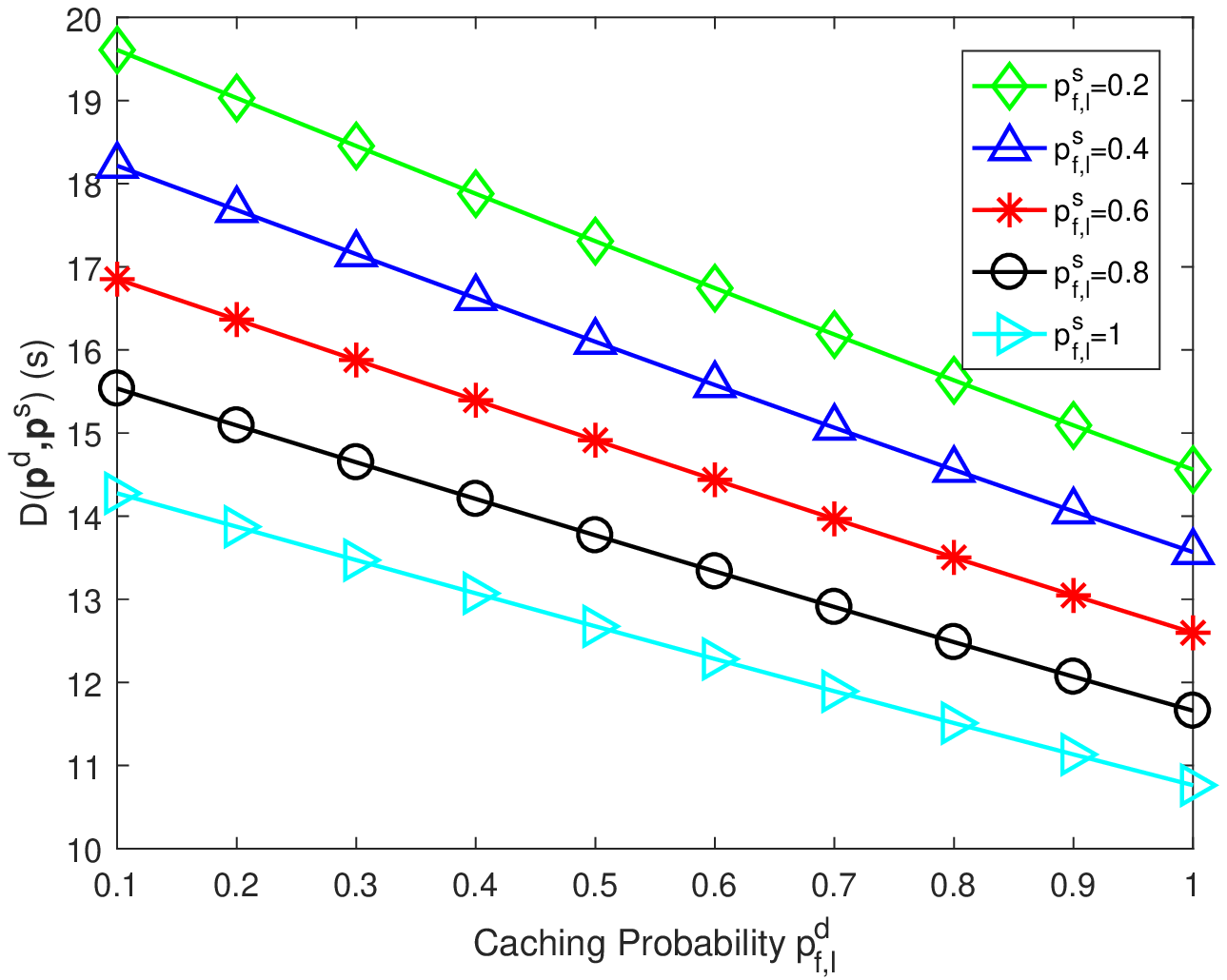}
\caption{The derived service delay versus the D2D caching probability.}
\end{center}
\end{figure}

\begin{figure}[t]
\begin{center}
\centering{}\includegraphics[scale=0.4]{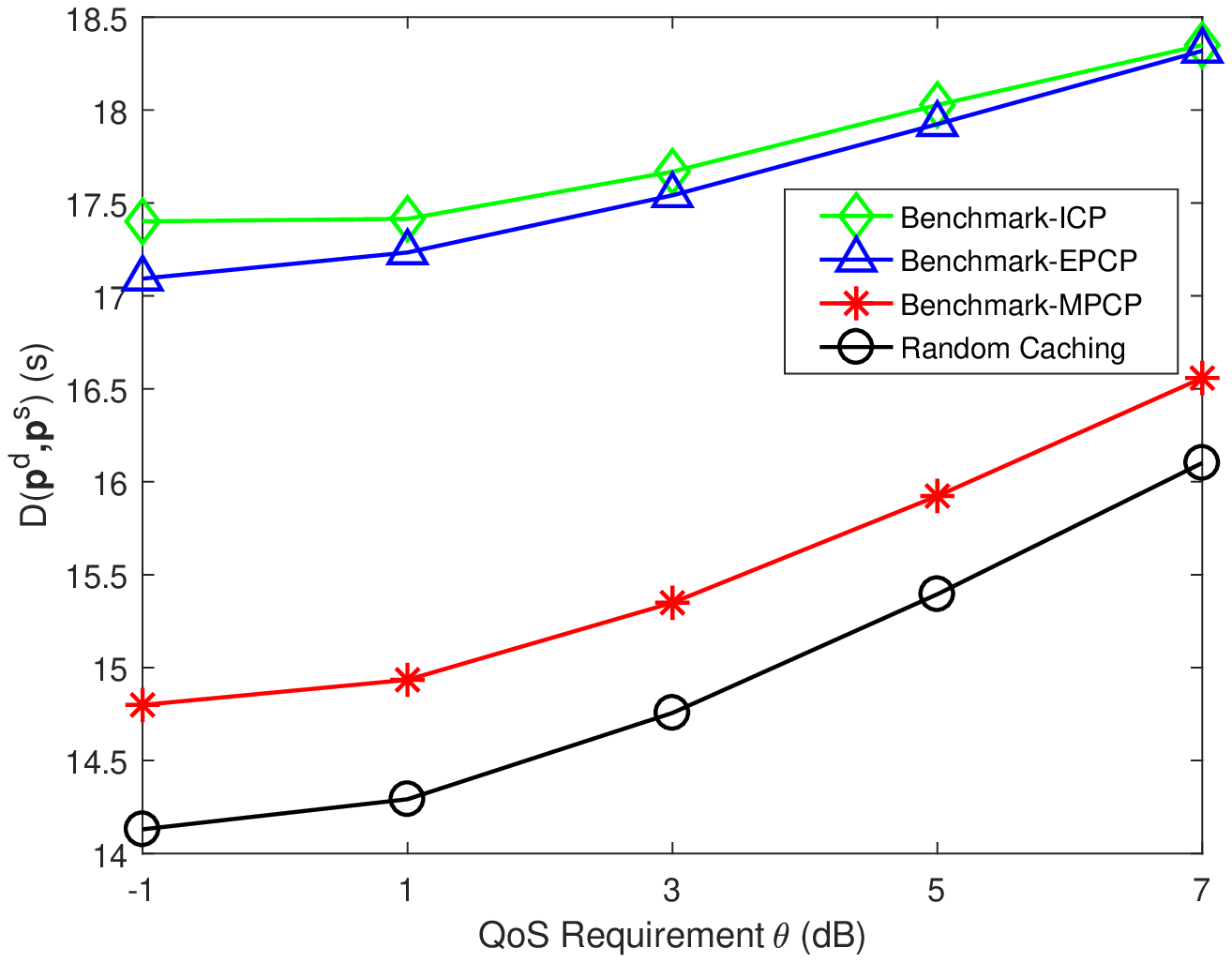}
\caption{The service delay versus the minimum QoS requirement.}
\end{center}
\end{figure}

\begin{figure}[t]
\centering{}\subfloat[The service delay versus the D2D cache size.]{\includegraphics[scale=0.4]{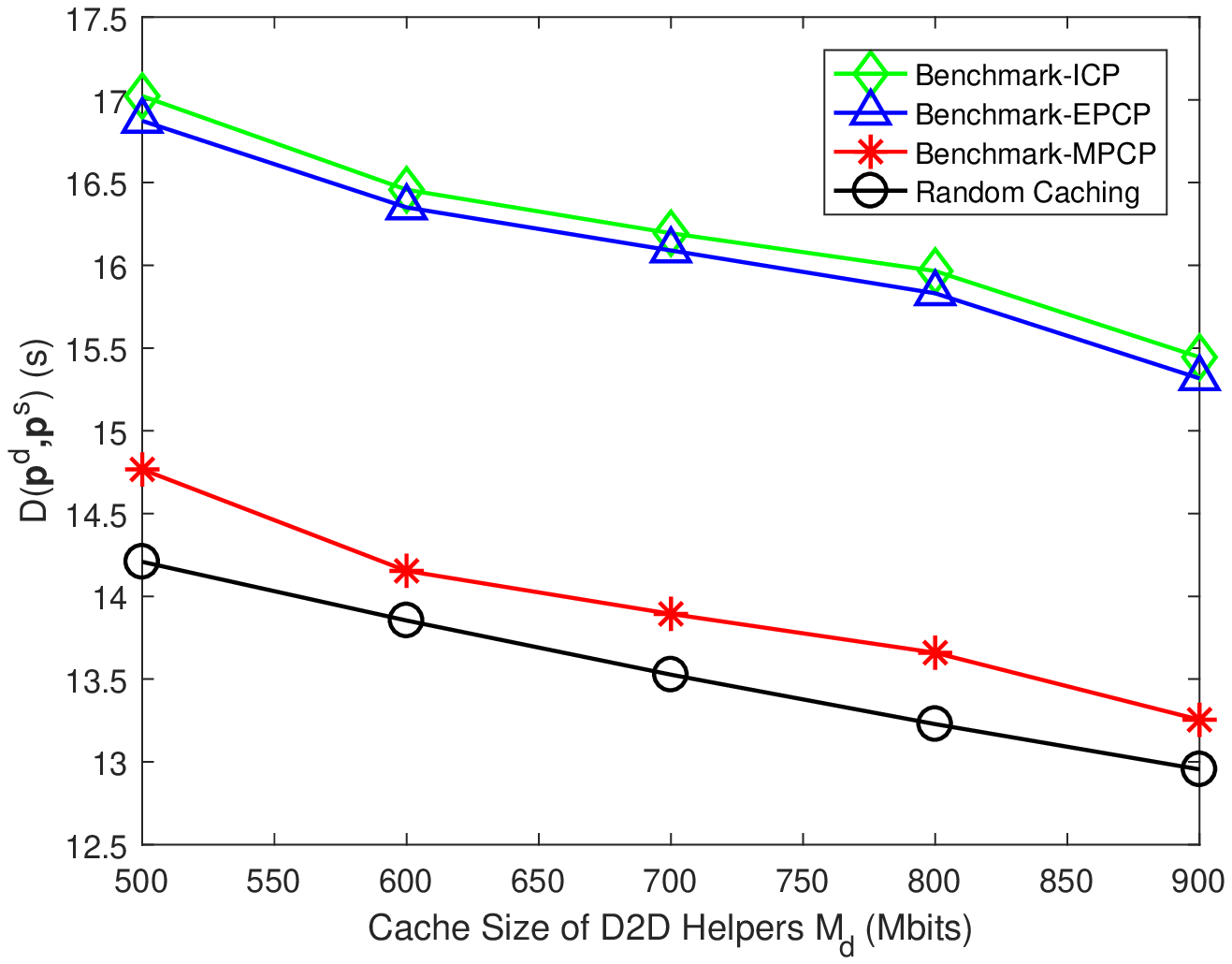}

}

\subfloat[The service delay versus the SBS cache size.]{\includegraphics[scale=0.4]{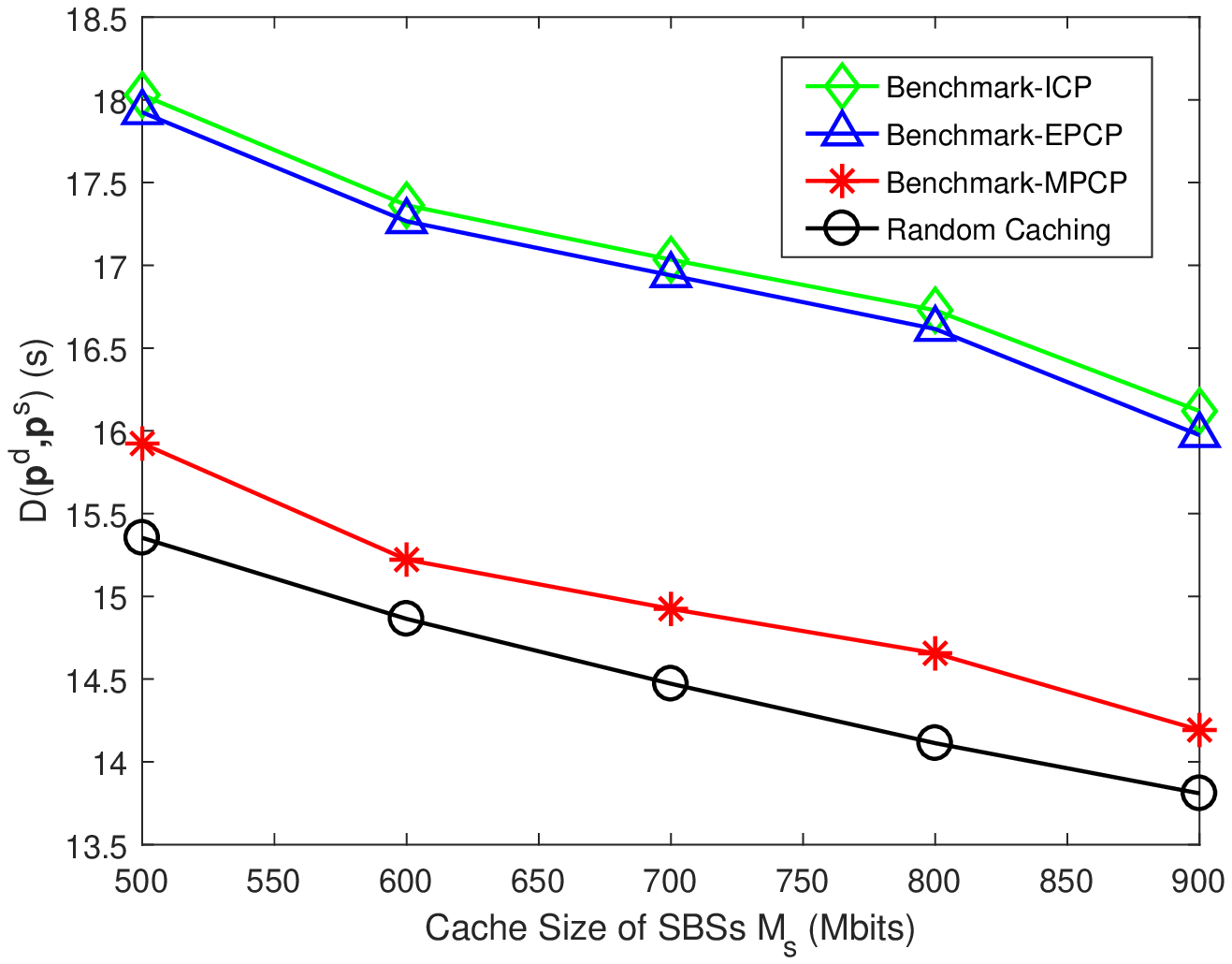}

}
\caption{The service delay versus different cache sizes of the D2D helper and SBS.}
\end{figure}

\begin{figure}[t]
\begin{center}
\centering{}\includegraphics[scale=0.4]{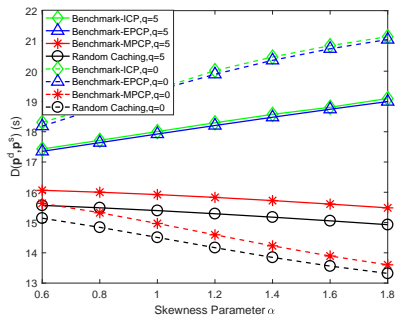}
\caption{The service delay versus the skewness parameter.}
\end{center}
\end{figure}

\begin{figure}[t]
\begin{center}
\centering{}\includegraphics[scale=0.4]{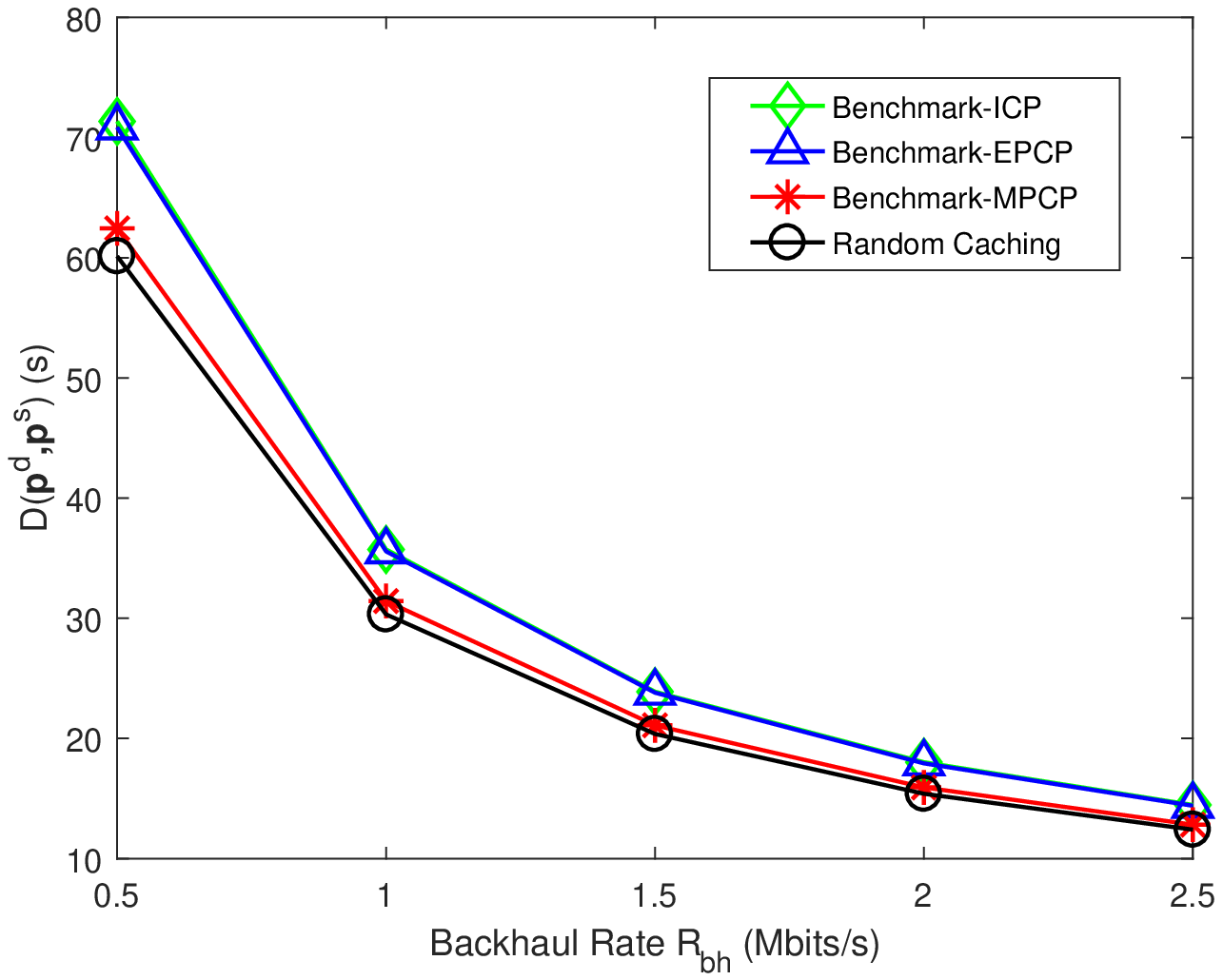}
\caption{The service delay versus the backhaul transmission rate.}
\end{center}
\end{figure}

The service delay in (\ref{eq:Delay}) is shown in Fig. 3,
varying with the caching probabilities.
We can see that larger values of random caching probabilities lead to lower service delay.
This is because larger caching probabilities indicate that
more content requests can be satisfied in the local caches of the D2D helpers and SBSs,
and the contents do not need to be retrieved via backhauls.
In practical implementations, though larger caching probabilities are beneficial for reducing the service delay,
the caching probabilities need to be carefully designed to meet the finite cache sizes.

In Fig. 4, the relationship between the service delay and the minimum QoS requirement $\theta$ is presented.
We can see that the service delay increases as $\theta$ grows.
The reason is as follows.
When $\theta$ grows, the successful transmission probabilities are substantially decreased, which leads to a higher service delay.
In addition, the proposed random caching scheme is superior to the three benchmark strategies in terms of service delay.
This is because random caching is able to fully exploit the accumulated cache sizes of D2D helpers and SBSs,
and more layers can be flexibly placed.
As for the three benchmarks, the MPCP shows better delay performance than EPCP and ICP,
since EPCP and ICP overlook the video popularity and the viewing quality preference.

In Fig. 5, we show the relationship between the service delay and the cache sizes $M_{\rm{d}}$ and $M_{\rm{s}}$.
It is clear that a larger cache size results in a lower service delay.
With a larger cache size, more requested super layers can be satisfied by the local caches of the D2D helpers and SBSs,
avoiding time-consuming content retrievals from the core network via backhauls.
Compared to Figs. 5(a) and 5(b),
we infer that allocating more cache sizes to the D2D helpers leads to more reduced service delay.
This is because the D2D helpers are geographically closer to the typical user,
and obtaining the requested layers from them is more delay-effective.

The relationship between the service delay and the skewness parameter $\alpha$ is provided in Fig. 6.
As $\alpha$ grows, the user requests increasingly focus on a small number of popular videos.
The requested layers are more likely to be locally stored, leading to a lower service delay.
There is an interesting phenomenon that, under the EPCP and ICP schemes, a larger $\alpha$ generates a larger service delay.
The reason is as follows.
The EPCP and ICP schemes ignore the actual video popularity and the viewing quality requirement,
resulting in indistinctive caching policies.
When the user requests concentrate on a small number of popular contents,
the two caching schemes fail to give priority to these popular videos.
More cache sizes are allocated to less popular videos than they are in the proposed caching scheme and MPLP,
leading to frequent backhaul retrievals of the popular files.
Therefore, the two schemes incur the increasing service delay as $\alpha$ grows.
As also shown in Fig. 6, when the plateau parameter $q$ decreases,
the delays of the proposed scheme and MPLP are impaired.
This is because the video popularity distribution is steeper,
and the user requests concentrate on the most popular videos.

In Fig. 7, we present the relationship between the service delay and the backhaul transmission rate $R_{\rm{bh}}$.
We see that a larger $R_{\rm{bh}}$ results in a lower service delay,
and the performance gap between the proposed caching scheme and MPLP decreases as $R_{\rm{bh}}$ grows.
This is due to the fact that, as $R_{\rm{bh}}$ grows,
the transmission rate of backhaul deliveries has an increasingly negligible impact on the service delay.
When $R_{\rm{bh}}$ is high enough, the backhaul capacity would be no longer a limiting factor in large-scale video distributions.

\begin{figure}[t]
\begin{center}
\centering{}\includegraphics[scale=0.4]{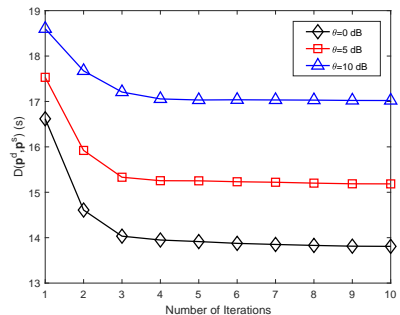}
\caption{The convergence property of the proposed algorithm.}
\end{center}
\end{figure}

In Fig. 8, the convergence property of the proposed algorithm is shown.
Obviously, Algorithm 1 can converge to the sub-optimal solution after a small number of iterations,
which validates the excellent convergence performance of the proposed algorithm.
Additionally, we can observe that larger values of the minimum QoS requirement $\theta$ result in larger service delays,
and in the meanwhile the conclusion derived from Fig. 4 is also clearly verified.
\section{Conclusion}
In this paper, we investigated the random caching scheme for delay minimization in D2D-assisted heterogeneous network.
To provide diversified viewing qualities of multimedia video services,
the super layers were transmitted to the user.
We firstly analyzed the successful transmission probabilities,
and then obtained the close-form expression for the overall service delay.
Based on this expression, we minimized the service delay efficiently by applying the improved standard gradient projection method.
Numerical results validate our analysis of successful transmission probabilities,
and the proposed random caching scheme was shown to be superior to the MPCP, EPCP and ICP strategies.
\begin{appendices}
\section{Proof of Theorem 1}
We start with the calculation of $P_{f,l}^{1}\left(\mathrm{SIR}_{\rm{d}}\geq\theta\right)$.
Keep in mind that, in Case 1, the intra-tier interference results from the D2D helpers that are farther away than $\rm{d}_{0}$ from the typical user.
Hence, $P_{f,l}^{1}\left(\mathrm{SIR}_{\rm{d}}\geq\theta\right)$ can be calculated as
\begin{gather}
P_{f,l}^{1}\left(\mathrm{SIR}_{\rm{d}}\geq\theta\right)=\int_{0}^{r_{\rm{c}}}f_{\rm{d}}(r)P_{f,l}^{1}\left(\mathrm{SIR}_{\rm{d}}\geq\theta|r_{0}^{\rm{d}}=r\right)\mathrm{d}r\label{eq:aa}
\end{gather}
where
\begin{gather}
f_{\rm{d}}(r)=\frac{2\pi\lambda_{\rm{d}}p_{f,l}^{\rm{d}}r\exp\left(-\lambda_{\rm{d}}p_{f,l}^{\rm{d}}\pi r^{2}\right)}{1-\exp\left(-\lambda_{\rm{d}}p_{f,l}^{\rm{d}}\pi r_{\rm{c}}^{2}\right)}\label{eq:a_1}
\end{gather}
is the probability density function (PDF) of the distance between the D2D helper $\rm{d}_{0}$ and the user.
For notational simplicity, the interference from other non-serving D2D helpers is $I_{\rm{d}_{1}}=\sum_{k\in\mathrm{\Phi_{\rm{d}}}\setminus \rm{d}_{0}}\left|h_{k}^{\rm{d}}\right|^{2}(r_{k}^{\rm{d}})^{-\alpha_{\rm{d}}}$.
The conditional successful transmission probability $P_{f,l}^{1}\left(\mathrm{SIR}_{\rm{d}}\geq\theta|r_{0}^{\rm{d}}=r\right)$ is calculated as
\begin{align}
 P_{f,l}^{1}&\left(\mathrm{SIR}_{\rm{d}}\geq\theta|r_{0}^{\rm{d}}=r\right)\nonumber\\
 &=P_{f,l}^{1}\left(\left|h_{0}^{\rm{d}}\right|^{2}\geq I_{\rm{d}_{1}}r^{\alpha_{\rm{d}}}\theta\right)\nonumber \\
 &\stackrel{\text{(a)}}{=}\mathbb{E}_{I_{\rm{d}_{1}}}\left[\exp\left(-\theta r^{\alpha_{\rm{d}}}I_{\rm{d}_{1}}\right)\right]\nonumber\\
  & =\mathbb{E}_{\mathrm{\Phi_{d}},h_{k}^{\rm{d}}}\left[\prod_{k\in\mathrm{\Phi_{\rm{d}}}\setminus \rm{d}_{0}}\exp\left(-\theta r^{\alpha_{\rm{d}}}\left|h_{k}^{\rm{d}}\right|^{2}\left(r_{k}^{\rm{d}}\right)^{-\alpha_{\rm{d}}}\right)\right]\nonumber\\
  & \stackrel{\text{(b)}}{=}\mathbb{E}_{\mathrm{\Phi_{d}}}\left[\prod_{k\in\mathrm{\Phi_{d}}\setminus \rm{d}_{0}}\frac{1}{1+\theta r^{\alpha_{\rm{d}}}\left(r_{k}^{\rm{d}}\right)^{-\alpha_{\rm{d}}}}\right]\nonumber\\
 & \stackrel{\text{(c)}}{=}\exp\left(-2\pi\lambda_{\rm{d}}\int_{r}^{\infty}\left(1-\frac{1}{1+\theta r^{\alpha_{\rm{d}}}\rho^{-\alpha_{\rm{d}}}}\right)\rho\mathrm{\rm{d}}\rho\right)\nonumber \\
 & =\exp\left(-\pi\lambda_{\rm{d}}\theta^{\frac{2}{\alpha_{\rm{d}}}}r^{2}G_{\alpha_{\rm{d}}}\left(\theta^{-\frac{2}{\alpha_{\rm{d}}}}\right)\right)\label{eq:a_2}
\end{align}
where $G_{a}(b)=\int_{b}^{\infty}\frac{1}{1+x^{\frac{a}{2}}}\mathrm{d}x$.
Since $\left|h_{0}^{\rm{d}}\right|^{2}$ follows the exponential distribution with unit mean,
its PDF is $f_{\left|h_{0}^{\rm{d}}\right|^{2}}(x)=\exp(-x)$, and then $P_{f,l}^{1}\left(\left|h_{0}^{\rm{d}}\right|^{2}\geq I_{\rm{d}_{1}}r^{\alpha_{\rm{d}}}\theta\right)$ can be accordingly calculated.
As $I_{\rm{d}_{1}}$ is another stochastic variable in $P_{f,l}^{1}\left(\left|h_{0}^{\rm{d}}\right|^{2}\geq I_{\rm{d}_{1}}r^{\alpha_{\rm{d}}}\theta\right)$, the expectation of $I_{\rm{d}_{1}}$ is supposed to be considered,
as shown in Step (a).
$\left|h_{k}^{\rm{d}}\right|^{2}$ also follows the exponential distribution with unit mean.
Therefore, its expectation can be easily yielded based on its known PDF, as given in Step (b).
Finally, Step (c) can be obtained by leveraging the probability generating functional (PGFL) property of the PPP,
and more details for this special property can refer to Definition 4.3 shown in \cite{Haenggi2012Stochastic}.
As a result, the expression for $P_{f,l}^{1}\left(\mathrm{SIR}_{\rm{d}}\geq\theta\right)$
can be obtained by substituting (\ref{eq:a_1}) and (\ref{eq:a_2}) into (\ref{eq:aa}),
given as follows
\begin{align}
 P_{f,l}^{1}&\left(\mathrm{SIR}_{\rm{d}}\geq\theta\right)\nonumber\\
 &=\int_{0}^{r_{\rm{c}}}\frac{2\pi p_{f,l}^{\rm{d}}\lambda_{\rm{d}}r}{1-\exp\left(-\lambda_{\rm{d}}p_{f,l}^{\rm{d}}\pi r_{\rm{c}}^{2}\right)}\nonumber\\
 &\exp\left(-\lambda_{\rm{d}}\pi r^{2}\left(p_{f,l}^{\rm{d}}+\theta^{\frac{2}{\alpha_{\rm{d}}}}G_{\alpha_{\rm{d}}}\left(\theta^{-\frac{2}{\alpha_{\rm{d}}}}\right)\right)\right)\mathrm{d}r\nonumber\\
 & =\int_{0}^{r_{\rm{c}}^2}\frac{\pi p_{f,l}^{\rm{d}}\lambda_{\rm{d}}}{1-\exp\left(-\lambda_{\rm{d}}p_{f,l}^{\rm{d}}\pi r_{\rm{c}}^{2}\right)}\nonumber\\
 &\exp\left(-\lambda_{\rm{d}}\pi r^{2}\left(p_{f,l}^{\rm{d}}+\theta^{\frac{2}{\alpha_{\rm{d}}}}G_{\alpha_{\rm{d}}}\left(\theta^{-\frac{2}{\alpha_{\rm{d}}}}\right)\right)\right)\mathrm{d}r^{2}\nonumber\\
 & =\frac{p_{f,l}^{\rm{d}}\left(1-\exp\left(-\lambda_{\rm{d}}\pi r^{2}\left(p_{f,l}^{\rm{d}}+\theta^{\frac{2}{\alpha_{\rm{d}}}}G_{\alpha_{\rm{d}}}\left(\theta^{-\frac{2}{\alpha_{\rm{d}}}}\right)\right)\right)\right)}{\left(1-\exp\left(-\lambda_{\rm{d}}p_{f,l}^{\rm{d}}\pi r_{\rm{c}}^{2}\right)\right)\left(p_{f,l}^{\rm{d}}+\theta^{\frac{2}{\alpha_{\rm{d}}}}G_{\alpha_{\rm{d}}}\left(\theta^{-\frac{2}{\alpha_{\rm{d}}}}\right)\right)}.\label{A_1}
\end{align}
From (\ref{A_1}), we obtain that when $p_{f,l}^{\rm{d}}=0$,
the denominator of $P_{f,l}^{1}\left(\mathrm{SIR}_{\rm{d}}\geq\theta\right)$ equals to zero.
In order to avoid this, it is stipulated that, when $p_{f,l}^{\rm{d}}=0$, $P_{f,l}^{1}\left(\mathrm{SIR}_{\rm{d}}\geq\theta\right)=0$ holds.
This is consistent with the practical situation,
since the successful transmission probability from the D2D helper equals to zero when there is no potential serving helper.
As a result, we can obtain the successful transmission probability in Case 1.

We proceed with Case 2.
The expression for $P_{f,l}^{2}\left(\mathrm{SIR}_{\rm{d}}\geq\theta\right)$ can be given by
\begin{gather}
P_{f,l}^{2}\left(\mathrm{SIR}_{\rm{d}}\geq\theta\right)=\int_{0}^{r_{\rm{c}}}f_{\rm{d}}(r)P_{f,l}^{2}\left(\mathrm{SIR}_{\rm{d}}\geq\theta|r_{0}^{\rm{d}}=r\right)\mathrm{d}r\label{eq:bb}
\end{gather}
\noindent where the interferences come from the D2D helpers
that are geographically closer than the D2D helper $\rm{d}_{0}$ and the helpers that are farther away than $\rm{d}_{0}$.
The interferences generated by the closer and farther SBSs are denoted by $I_{\rm{d}_{2}}$ and $I_{\rm{d}_{3}}$, respectively.
According to the above analysis, $P_{f,l}^{2}\left(\mathrm{SIR}_{\rm{d}}\geq\theta|r_{0}^{\rm{d}}=r\right)$ can be obtained by
\begin{align}
 P_{f,l}^{2}&\left(\mathrm{SIR}_{\rm{d}}\geq\theta|r_{0}^{\rm{d}}=r\right)\nonumber\\
 & =P_{f,l}^{2}\left(\left|h_{0}^{\rm{d}}\right|^{2}r^{-\alpha_{\rm{d}}}\geq\left(I_{\rm{d}_{2}}+I_{\rm{d}_{3}}\right)\theta\right)\nonumber\\
  & =\mathbb{E}_{I_{\rm{d}_{2},}I_{\rm{d}_{3}}}\left[\exp\left(-\theta r^{\alpha_{\rm{d}}}\left(I_{\rm{d}_{2}}+I_{\rm{d}_{3}}\right)\right)\right]\nonumber\\
 & =\mathcal{L}_{I_{\rm{d}_{2}}}\left(\theta r^{\alpha_{\rm{d}}}\right)\mathcal{L}_{I_{\rm{d}_{3}}}\left(\theta r^{\alpha_{\rm{d}}}\right)
\end{align}
where  $\mathcal{L}_{I_{\rm{d}_{2}}}\left(\theta r^{\alpha_{\rm{d}}}\right)$ and $\mathcal{L}_{I_{\rm{d}_{3}}}\left(\theta r^{\alpha_{\rm{d}}}\right)$ are the Laplace transforms
regarding $I_{\rm{d}_{2}}$ and $I_{\rm{d}_{3}}$, respectively.
The Laplace transforms
$L_{I_{\rm{d}_{2}}}\left(\theta r^{\alpha_{\rm{d}}}\right)$ and $L_{I_{\rm{d}_{3}}}\left(\theta r^{\alpha_{\rm{d}}}\right)$ are given by
\begin{align}
&\mathcal{L}_{I_{\rm{d}_{2}}}\left(\theta r^{\alpha_{\rm{d}}}\right)=\exp\left(-\pi\lambda_{\rm{d}}\theta^{\frac{2}{\alpha_{\rm{d}}}}r^{2}\left(G_{\alpha_{\rm{d}}}(0)-G_{\alpha_{\rm{d}}}\left(\theta^{-\frac{2}{\alpha_{\rm{d}}}}\right)\right)\right),\label{eq:b_1}
\end{align}
\begin{align}
&\mathcal{L}_{I_{\rm{d}_{3}}}\left(\theta r^{\alpha_{\rm{d}}}\right)=\exp\left(-\pi\lambda_{\rm{d}}\theta^{\frac{2}{\alpha_{\rm{d}}}}r^{2}G_{\alpha_{\rm{d}}}\left(\theta^{-\frac{2}{\alpha_{\rm{d}}}}\right)\right).\label{eq:b_2}
\end{align}
\noindent Based on (\ref{eq:b_1}) and (\ref{eq:b_2}), $P_{f,l}^{2}\left(\mathrm{SIR}_{\rm{d}}\geq\theta\right)$ can be obtained by
\begin{align}
 P_{f,l}^{2}&\left(\mathrm{SIR}_{\rm{d}}\geq\theta\right)\nonumber\\
 & =\int_{0}^{r_{\rm{c}}}\frac{2\pi p_{f,l}^{\rm{d}}\lambda_{\rm{d}}r}{1-\exp\left(-\lambda_{\rm{d}}p_{f,l}^{\rm{d}}\pi r_{\rm{c}}^{2}\right)}\nonumber\\
 &\exp\left(-\lambda_{\rm{d}}\pi r^{2}\left(p_{f,l}^{\rm{d}}+\theta^{\frac{2}{\alpha_{\rm{d}}}}G_{\alpha_{\rm{d}}}(0)\right)\right)\mathrm{d}r\nonumber\\
 & =\int_{0}^{r_{\rm{c}}^{2}}\frac{\pi p_{f,l}^{\rm{d}}\lambda_{\rm{d}}}{1-\exp\left(-\lambda_{\rm{d}}p_{f,l}^{\rm{d}}\pi r_{\rm{c}}^{2}\right)}\nonumber\\
 &\exp\left(-\lambda_{\rm{d}}\pi r^{2}\left(p_{f,l}^{\rm{d}}+\theta^{\frac{2}{\alpha_{\rm{d}}}}G_{\alpha_{\rm{d}}}(0)\right)\right)\mathrm{d}r^{2}\nonumber\\
 & =\frac{p_{f,l}^{\rm{d}}\left(1-\exp\left(-\lambda_{\rm{d}}\pi r^{2}\left(p_{f,l}^{\rm{d}}+\theta^{\frac{2}{\alpha_{\rm{d}}}}G_{\alpha_{\rm{d}}}(0)\right)\right)\right)}{\left(1-\exp\left(-\lambda_{\rm{d}}p_{f,l}^{\rm{d}}\pi r_{c}^{2}\right)\right)\left(p_{f,l}^{\rm{d}}+\theta^{\frac{2}{\alpha_{\rm{d}}}}G_{\alpha_{\rm{d}}}(0)\right)}\nonumber.
\end{align}
\noindent On the other hand, when $P_{f,l}^{\rm{d}}=0$, we set $P_{f,l}^{2}\left(\mathrm{SIR}_{\rm{d}}\geq\theta\right)$=0.
From the above analysis, the proof of Theorem 1 is complete.
\QEDA
\section{Proof of Theorem 3}
When the user is served by the nearest MBS, the successful transmission probability $P\left(\mathrm{SIR}_{\rm{m}}\geq\theta\right)$ is calculated as\\
\begin{gather}
P\left(\mathrm{SIR}_{\rm{m}}\geq\theta\right)=\int_{0}^{\infty}f_{\rm{m}}(r)P\left(\mathrm{SIR}_{\rm{m}}\geq\theta|r_{0}^{\rm{m}}=r\right)\mathrm{d}r\label{eq:dd}
\end{gather}
where
\begin{gather}
f_{\rm{m}}(r)=2\pi\lambda_{\rm{m}}r\exp\left(-\lambda_{\rm{m}}\pi r^{2}\right).
\end{gather}
We denote the interference from other non-serving MBSs as
$I_{\rm{m}}=\sum_{k\in\mathrm{\Phi_{m}\setminus m_{0}}}\left|h_{k}^{\rm{m}}\right|^{2}\left(r_{k}^{\rm{m}}\right)^{-\alpha_{\rm{m}}}$. $P(\mathrm{SIR}_{\rm{m}}\geq\theta|r_{0}^{\rm{m}}=r)$ is calculated as
\begin{align}
P&\left(\mathrm{SIR}_{\rm{m}}\geq\theta|r_{0}^{\rm{m}}=r\right)=\mathcal{L}_{I_{\rm{m}}}\left(\theta r^{\alpha_{\rm{m}}}\right)\nonumber
\end{align}
\begin{align}
&=\exp\left(-\pi\lambda_{\rm{m}}r^{2}\theta^{\frac{2}{\alpha_{\rm{m}}}}G_{\alpha_{\rm{m}}}\left(\theta^{-\frac{2}{\alpha_{\rm{m}}}}\right)\right). \label{eq:dd_1}
\end{align}
\noindent By substituting (\ref{eq:dd_1}) into (\ref{eq:dd}), we can obtain
\begin{align}
 P&\left(\mathrm{SIR}_{\rm{m}}\geq\theta\right)\nonumber\\
 & =2\pi\lambda_{\rm{m}}\int_{0}^{\infty}r\exp\left(-\pi\lambda_{\rm{m}}r^{2}\theta^{\frac{2}{\alpha_{\rm{m}}}}G_{\alpha_{\rm{m}}}\left(\theta^{-\frac{2}{\alpha_{\rm{m}}}}\right)\right)\mathrm{d}r\nonumber\\
 & =\pi\lambda_{\rm{m}}\int_{0}^{\infty}\exp\left(-\pi\lambda_{\rm{m}}r^{2}\left(1+\theta^{\frac{2}{\alpha_{\rm{m}}}}G_{\alpha_{m}}\left(\theta^{-\frac{2}{\alpha_{\rm{m}}}}\right)\right)\right)\mathrm{d}r^{2}\nonumber\\
 &=\left[1+\theta^{\frac{2}{\alpha_{\rm{m}}}}G_{\alpha_{\rm{m}}}\left(\theta^{-\frac{2}{\alpha_{\rm{m}}}}\right)\right]^{-1}.
\end{align}
As a result, the expression for $P\left(\mathrm{SIR}_{\rm{m}}\geq\theta\right)$ is obtained.
\QEDA
%
\end{appendices}
\bibliographystyle{IEEEtran}

\begin{thebibliography}{10}

\bibitem{Gupta2012H}
R.~Gupta, A.~Pulipaka, P.~Seeling, L.~J. Karam, and M.~Reisslein, ``{H.264}
  coarse grain scalable {(CGS)} and medium grain scalable {(MGS)} encoded
  video: A trace based traffic and quality evaluation,'' \emph{IEEE Trans.
  Broadcast.}, vol.~58, no.~3, pp. 428--439, Sep. 2012.

\bibitem{Zhang2018Energy2}
R.~Zhang, Y.~Li, C.~X. Wang, Y.~Ruan, Y.~Fu, and H.~Zhang, ``Energy-spectral
  efficiency trade-off in underlaying mobile {D2D} communications: {An}
  economic efficiency perspective,'' \emph{IEEE Trans. Wireless Commun.},
  vol.~17, no.~7, pp. 4288--4301, Jul. 2018.

\bibitem{2016Ericsson}
Ericsson, ``Ericsson mobility report,''
  \url{http://www.ericsson.com/en/mobility-report/reports/november-2017}, Nov.
  2017.

\bibitem{Zhang2018Near}
X.~Zhang, T.~Lv, and S.~Yang, ``Near-optimal layer placement for scalable
  videos in cache-enabled small-cell networks,'' \emph{IEEE Trans. Veh. Tech.},
  vol.~67, no.~9, pp. 9047--9051, Sep. 2018.

\bibitem{Tao2015Content}
M.~Tao, E.~Chen, H.~Zhou, and W.~Yu, ``Content-centric sparse multicast
  beamforming for cache-enabled cloud {RAN},'' \emph{IEEE Trans. Wireless
  Commun.}, vol.~15, no.~9, pp. 6118--6131, Sep. 2016.

\bibitem{Chen2017Cooperative}
Z.~Chen, J.~Lee, T.~Q.~S. Quek, and M.~Kountouris, ``Cooperative caching and
  transmission design in cluster-centric small cell networks,'' \emph{IEEE
  Trans. Wireless Commun.}, vol.~16, no.~5, pp. 3401--3415, May 2017.

\bibitem{Zhang2017Multicast}
X.~Zhang, H.~Gao, and T.~Lv, ``Multicast beamforming for scalable videos in
  cache-enabled heterogeneous networks,'' in \emph{Proc. IEEE Wireless Commun.
  Networking Conf. (WCNC)}, San Francisco, CA, Mar. 2017, pp. 1--6.

\bibitem{Xu2017Fundamental}
F.~Xu, M.~Tao, and K.~Liu, ``Fundamental tradeoff between storage and latency
  in cache-aided wireless interference networks,'' \emph{IEEE Trans. Inf.
  Theory}, vol.~63, no.~11, pp. 7464--7491, Jun. 2017.

\bibitem{Lampiris2018Adding}
E.~Lampiris and P.~Elia, ``Adding transmitters dramatically boosts
  coded-caching gains for finite file sizes,'' \emph{IEEE J. Sel. Areas
  Commun.}, vol.~36, no.~6, pp. 1176--1188, Jun. 2018.

\bibitem{Wen2017Random}
W.~Wen, Y.~Cui, F.~C. Zheng, S.~Jin, and Y.~Jiang, ``Random caching based
  cooperative transmission in heterogeneous wireless networks,'' \emph{IEEE
  Trans. Commun.}, vol.~66, no.~7, pp. 2809--2825, Jul. 2018.

\bibitem{cui2016analysis}
Y.~Cui and D.~Jiang, ``Analysis and optimization of caching and multicasting in
  large-scale cache-enabled heterogeneous wireless networks,'' \emph{IEEE
  Trans. Wireless Commun.}, vol.~15, no.~7, pp. 5101--5112, Jul. 2016.

\bibitem{Zheng2017Probabilistic}
C.~Zheng, N.~Pappas, and M.~Kountouris, ``Probabilistic caching in wireless
  {D2D} networks: Cache hit optimal versus throughput optimal,'' \emph{IEEE
  Commun. Lett.}, vol.~21, no.~3, pp. 584--587, Mar. 2017.

\bibitem{Zhang2018Energy}
X.~Zhang, T.~Lv, W.~Ni, J.~M. Cioffi, N.~C. Beaulieu, and Y.~J. Guo,
  ``Energy-efficient caching for scalable videos in heterogeneous networks,''
  \emph{IEEE J. Sel. Areas Commun.}, vol.~36, no.~8, pp. 1802--1815, Aug. 2018.

\bibitem{Wu2017Energy}
Q.~Wu, G.~Y. Li, W.~Chen, and D.~W.~K. Ng, ``Energy-efficient {D2D} overlaying
  communications with spectrum-power trading,'' \emph{IEEE Trans. Wireless
  Commun.}, vol.~16, no.~7, Jul. 2017.

\bibitem{Ahmed2018A}
M.~Ahmed, Y.~Li, M.~Waqas, M.~Sheraz, D.~Jin, and Z.~Han, ``A survey on
  socially aware device-to-device communications,'' \emph{IEEE Commun. Surveys
  Tuts.}, vol.~20, no.~3, pp. 2169--2197, Thirdquarter 2018.

\bibitem{chen2016cache}
B.~Chen, C.~Yang, and A.~F. Molisch, ``Cache-enabled device-to-device
  communications: Offloading gain and energy cost,'' \emph{IEEE Trans. Wireless
  Commu.}, vol.~16, no.~7, pp. 4519--4536, Jul. 2017.

\bibitem{Chen2018Caching}
B.~Chen and C.~Yang, ``Caching policy for cache-enabled {D2D} communications by
  learning user preference,'' \emph{IEEE Trans. Commun.}, vol.~66, no.~12, pp.
  6586--6601, Dec. 2018.

\bibitem{Deng2018The}
N.~Deng and M.~Haenggi, ``The benefits of hybrid caching in {Gauss-Poisson D2D}
  networks,'' \emph{IEEE J. Sel. Areas Commun.}, vol.~36, no.~8, pp.
  1217--1230, Aug. 2018.

\bibitem{Zhang2016Efficient}
L.~Zhang, M.~Xiao, G.~Wu, and S.~Li, ``Efficient scheduling and power
  allocation for {D2D}-assisted wireless caching networks,'' \emph{IEEE Trans.
  Commun.}, vol.~64, no.~6, pp. 2438--2452, Jun. 2016.

\bibitem{Schwarz2007Overview}
H.~Schwarz, D.~Marpe, and T.~Wiegand, ``Overview of the scalable video coding
  extension of the {H.264/AVC} standard,'' \emph{IEEE Trans. Circuits Syst.
  Video Technol.}, vol.~17, no.~9, pp. 1103--1120, Sep. 2007.

\bibitem{Guo2018Multi}
C.~Guo, Y.~Cui, D.~W.~K. Ng, and Z.~Liu, ``Multi-quality multicast beamforming
  with scalable video coding,'' \emph{IEEE Trans. Commun.}, vol.~66, no.~11,
  pp. 5662--5677, Nov. 2018.

\bibitem{Ostovari2015Scalable}
P.~Ostovari, J.~Wu, A.~Khreishah, and N.~B. Shroff, ``Scalable video streaming
  with helper nodes using random linear network coding,'' \emph{IEEE/ACM Trans.
  Netw.}, vol.~24, no.~3, pp. 1574--1587, Jun. 2015.

\bibitem{Zhan2018SVC}
C.~{Zhan} and G.~{Yao}, ``{SVC} video delivery in cache-enabled wireless
  {HetNet},'' \emph{IEEE Syst. J.}, vol.~12, no.~4, pp. 3885--3888, Dec. 2018.

\bibitem{Zhang2017Layered}
Z.~{Zhang}, D.~{Liu}, and Y.~{Yuan}, ``Layered hierarchical caching for
  {SVC}-based {HTTP} adaptive streaming over {C-RAN},'' in \emph{Proc. IEEE
  Wireless Commun. Networking Conf. (WCNC)}, San Frnacisco,CA, Mar. 2017, pp.
  1--6.

\bibitem{Chen2015Delay}
Z.~Chen, L.~Qiu, Y.~Jin, and X.~Liang, ``Delay-aware uplink user association
  and power control in heterogeneous cellular networks,'' \emph{IEEE Wireless
  Commun. Lett.}, vol.~4, no.~6, pp. 661--664, Dec. 2015.

\bibitem{Amer2018Inter}
R.~Amer, M.~M. Butt, M.~Bennis, and N.~Marchetti, ``Inter-cluster cooperation
  for wireless {D2D} caching networks,'' \emph{IEEE Trans. Wireless Commun.},
  vol.~17, no.~9, pp. 6108--6121, Jul. 2018.

\bibitem{Amer2018Optimizing}
R.~Amer, H.~Elsawy, M.~M. Butt, E.~A. Jorswieck, M.~Bennis, and N.~Marchetti,
  ``Optimizing joint probabilistic caching and communication for clustered
  {D2D} networks,'' \emph{in arXiv}, 2018.

\bibitem{Li2018Learning}
Y.~Li, C.~Zhong, M.~C. Gursoy, and S.~Velipasalar, ``Learning-based delay-aware
  caching in wireless {D2D} caching networks,'' \emph{IEEE Access}, vol.~6, pp.
  77\,250--77\,264, Nov. 2018.

\bibitem{Wang2016Joint}
Y.~Wang, X.~Tao, X.~Zhang, and G.~Mao, ``Joint caching placement and user
  association for minimizing user download delay,'' \emph{IEEE Access}, vol.~4,
  pp. 8625--8633, 2016.

\bibitem{Zhang2016Socially}
G.~Zhang, K.~Yang, and H.~Chen, ``Socially aware cluster formation and radio
  resource allocation in {D2D} networks,'' \emph{IEEE Wireless Commun.},
  vol.~23, no.~4, pp. 68--73, Aug. 2016.

\bibitem{Zhang2017Cost}
S.~Zhang, N.~Zhang, P.~Yang, and X.~Shen, ``Cost-effective cache deployment in
  mobile heterogeneous networks,'' \emph{IEEE Trans. Veh. Tech.}, vol.~66,
  no.~12, pp. 11\,264--11\,276, Dec. 2017.

\bibitem{Wildemeersch2014Successive}
M.~Wildemeersch, T.~Q.~S. Quek, M.~Kountouris, A.~Rabbachin, and C.~H. Slump,
  ``Successive interference cancellation in heterogeneous networks,''
  \emph{IEEE Trans. Commun.}, vol.~62, no.~12, pp. 4440--4453, Dec. 2014.

\bibitem{Xu2017Modeling}
X.~Xu and M.~Tao, ``Modeling, analysis, and optimization of coded caching in
  small-cell networks,'' \emph{IEEE Trans. Commun.}, vol.~65, no.~8, pp.
  3415--3428, Aug. 2017.

\bibitem{Song2017Control}
H.~{Song}, X.~{Fang}, L.~{Yan}, and Y.~{Fang}, ``Control/user plane decoupled
  architecture utilizing unlicensed bands in {LTE} systems,'' \emph{IEEE
  Wireless Commun.}, vol.~24, no.~5, pp. 132--142, Oct. 2017.

\bibitem{Zhang2019Economical}
X.~Zhang, T.~Lv, Y.~Ren, W.~Ni, N.~C. Beaulieu, and Y.~J. Guo, ``Economical
  caching for scalable videos in cache-enabled heterogeneous networks,''
  \emph{IEEE J. Sel. Areas Commun.}, vol.~37, no.~7, pp. 4440--4453, May 2018.

\bibitem{breslau1999web}
L.~Breslau, P.~Cao, L.~Fan, G.~Phillips, and S.~Shenker, ``Web caching and
  {Z}ipf-like distributions: Evidence and implications,'' in \emph{Proc. IEEE
  Conf. on Computer Commun. (INFOCOM)}, New York, USA, Mar. 1999, pp. 126--134.

\bibitem{Lee2019Throughput}
M.~Lee, M.~Ji, A.~Molisch, and N.~Sastry, ``Throughput-outage analysis and
  evaluation of cache-aided {D2D} networks with measured popularity
  distributions,'' \emph{in arXiv}, Feb. 2019.

\bibitem{Wu2016Caching}
L.~Wu and W.~Zhang, ``Caching-based scalable video transmission over cellular
  networks,'' \emph{IEEE Commun. Lett.}, vol.~20, no.~6, pp. 1156--1159, Jun.
  2016.

\bibitem{Wu2018Energy}
H.~Wu and H.~Lu, ``Energy and delay optimization for cache-enabled dense small
  cell networks,'' \emph{in arXiv}, Mar. 2018.

\bibitem{GriewankEvaluating}
A.~Griewank and A.~Walther, \emph{Evaluating Derivatives Principles and
  Techniques of Algorithmic Differentiation, Frontiers in Applied
  Mathematics}.\hskip 1em plus 0.5em minus 0.4em\relax Philadephia, 2000.

\bibitem{Zhang2008Computing}
H.~Zhang, X.~Yi, C.~Zhang, and L.~Dong, ``Computing the high order derivatives
  with automatic differentiation and its application in {Chebyshev's} method,''
  in \emph{Proc. Fourth International Conference on Natural Computation
  (ICNC)}, Jinan, China, Oct. 2008, pp. 1--5.

\bibitem{Song2017OptimalContent}
J.~{Song}, H.~{Song}, and W.~{Choi}, ``Optimal content placement for wireless
  femto-caching network,'' \emph{IEEE Trans. Wireless Commun.}, vol.~16, no.~7,
  pp. 4433--4444, Jul. 2017.

\bibitem{Muller2016Smart}
S.~{Muller}, O.~{Atan}, M.~{van der Schaar}, and A.~{Klein}, ``Smart caching in
  wireless small cell networks via contextual multi-armed bandits,'' in
  \emph{Proc. IEEE Int. Conf. Commun. (ICC)}, Kuala Lumpur, Malaysia, May 2016,
  pp. 1--7.

\bibitem{Xiang2018Secure}
L.~{Xiang}, D.~W.~K. {Ng}, R.~{Schober}, and V.~W.~S. {Wong}, ``Secure video
  streaming in heterogeneous small cell networks with untrusted cache
  helpers,'' \emph{IEEE Trans. Wireless Commun.}, vol.~17, no.~4, pp.
  2645--2661, Apr. 2018.

\bibitem{Zhang2020Double}
Z.~{Zhang}, H.~{Chen}, M.~{Hua}, C.~{Li}, Y.~{Huang}, and L.~{Yang}, ``Double
  coded caching in ultra dense networks: Caching and multicast scheduling via
  deep reinforcement learning,'' \emph{IEEE Trans. Commun.}, vol.~68, no.~2,
  pp. 1071--1086, Feb. 2020.

\bibitem{Yang2016Analysis}
C.~Yang, Y.~Yao, Z.~Chen, and B.~Xia, ``Analysis on cache-enabled wireless
  heterogeneous networks,'' \emph{IEEE Trans. Wireless Commun.}, vol.~15,
  no.~1, pp. 131--145, Jan. 2016.

\bibitem{Afshang2016Modeling}
M.~{Afshang}, H.~S. {Dhillon}, and P.~H. {Joo Chong}, ``Modeling and
  performance analysis of clustered device-to-device networks,'' \emph{IEEE
  Trans. Wireless Commun.}, vol.~15, no.~7, pp. 4957--4972, Jul. 2016.

\bibitem{Haenggi2012Stochastic}
M.~Haenggi, \emph{Stochastic Geometry for Wireless Networks}.\hskip 1em plus
  0.5em minus 0.4em\relax Cambridge University Press, 2012.

\end{thebibliography}

\begin{IEEEbiography}[{\includegraphics[width=1in,height=1.25in,clip,keepaspectratio]{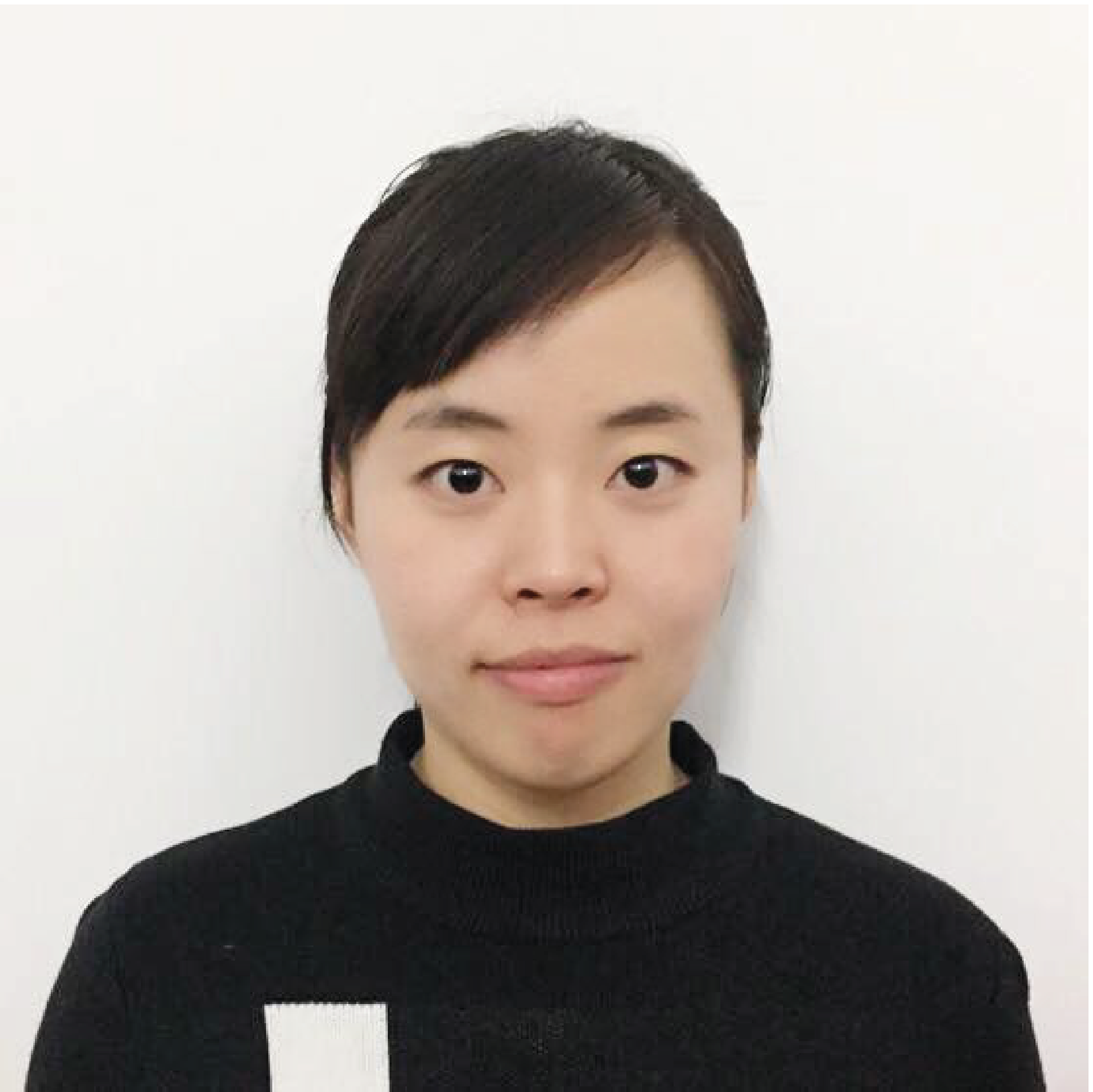}}]{Xuewei Zhang}
received the Ph.D. degree in information and communication information engineering
from the Beijing University of Posts and Telecommunications (BUPT), Beijing, China, in 2020.
She is currently an Associate Professor with the School of Communication and Information Engineering,
Xi'an University of Posts and Telecommunications (XUPT), Xi'an, China.
Her research interests include wireless edge caching, resource allocation and heterogeneous networks.
\end{IEEEbiography}
\begin{IEEEbiography}[{\includegraphics[width=1in,height=1.25in,clip,keepaspectratio]{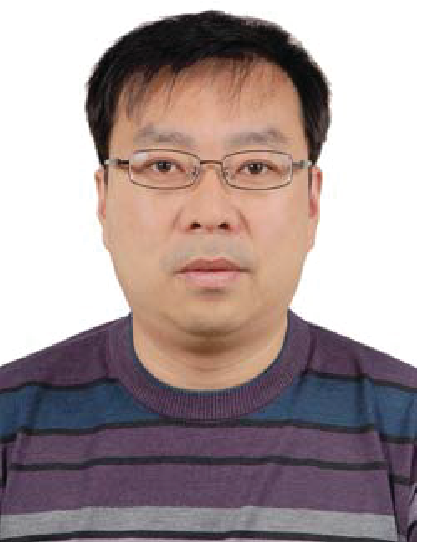}}]{Tiejun Lv}
(M'08-SM'12) received the M.S. and Ph.D. degrees in electronic engineering from the University of Electronic Science and Technology of China (UESTC), Chengdu, China, in 1997 and 2000, respectively. From January 2001 to January 2003, he was a Postdoctoral Fellow with Tsinghua University, Beijing, China. In 2005, he was promoted to a Full Professor with the School of Information and Communication Engineering, Beijing University of Posts and Telecommunications (BUPT). From September 2008 to March 2009, he was a Visiting Professor with the Department of Electrical Engineering, Stanford University, Stanford, CA, USA. He is the author of 3 books, more than 80 published IEEE journal papers and 180 conference papers on the physical layer of wireless mobile communications. His current research interests include signal processing, communications theory and networking. He was the recipient of the Program for New Century Excellent Talents in University Award from the Ministry of Education, China, in 2006. He received the Nature Science Award in the Ministry of Education of China for the hierarchical cooperative communication theory and technologies in 2015.
\end{IEEEbiography}
\begin{IEEEbiography}[{\includegraphics[width=1in,height=1.25in,clip,keepaspectratio]{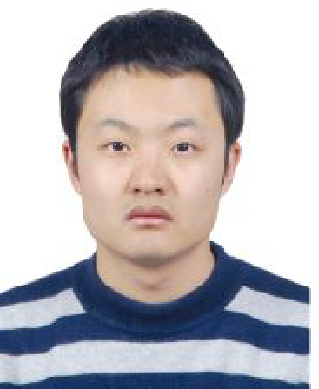}}]{Yuan Ren}
received the B.Eng. degree in information engineering and the Ph.D. degree in signal and information processing
from the Beijing University of Posts and Telecommunications (BUPT), Beijing, China, in 2010 and 2017, respectively.
He is currently a Lecturer with the School of Communication and Information Engineering,
Xi'an University of Posts and Telecommunications (XUPT), Xi'an, China.
His current research interests include green communications, wireless caching, and cooperative communications.
\end{IEEEbiography}
\begin{IEEEbiography}[{\includegraphics[width=1in,height =1.25in,clip,keepaspectratio]{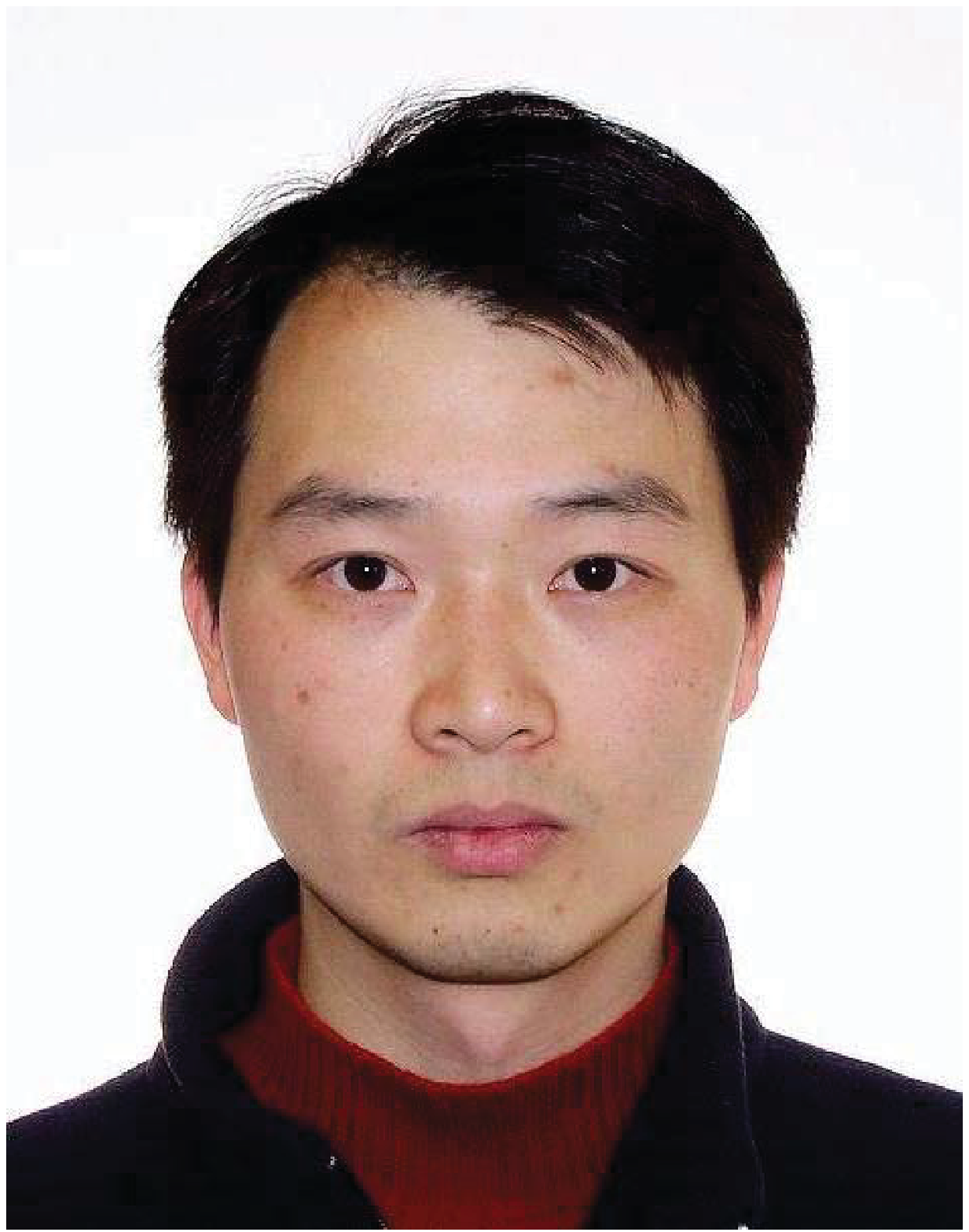}}]{Wei Ni} (M'09-SM'15) received the B.E. and Ph.D. degrees in Electronic Engineering from Fudan University, Shanghai, China, in 2000 and 2005, respectively. Currently,he is a Group Leader and Principal Research Scientist at CSIRO, Sydney, Australia, and an adjunct professor at the University of Technology Sydney and an Honorary Professor at Macquarie University, Sydney. Prior to this, he was a Postdoctoral Research Fellow at Shanghai Jiaotong University from 2005--2008; Deputy Project Manager at the Bell Labs, Alcatel/Alcatel-Lucent from 2005--2008; and Senior Researcher at Devices R\&D, Nokia from 2008--2009. His research interests include signal processing, stochastic optimization, as well as their applications to network efficiency and integrity.

Dr Ni is the Chair of IEEE Vehicular Technology Society (VTS) New South Wales (NSW) Chapter since 2020 and an Editor of IEEE Transactions on Wireless Communications since 2018. He served first the Secretary and then Vice-Chair of IEEE NSW VTS Chapter from 2015--2019, Track Chair for VTC-Spring 2017, Track Co-chair for IEEE VTC-Spring 2016, Publication Chair for BodyNet 2015, and  Student Travel Grant Chair for WPMC 2014.
\end{IEEEbiography}
\begin{IEEEbiography}[{\includegraphics[width=1in,height=1.25in,clip,keepaspectratio]{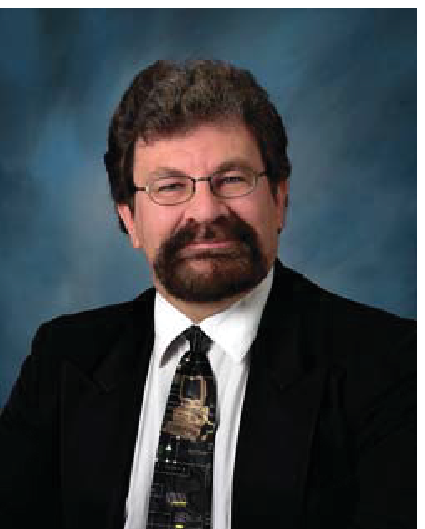}}]{Norman C. Beaulieu}
is an Academician of the Royal Society of Canada (RSC), and an Academician of the Canadian Academy of Engineering (CAE).
He is a Fellow of the Institute of Electrical and Electronics Engineers (IEEE),
a Fellow of the Institution of Engineering and Technology (IET) of the United Kingdom,
a Fellow of the Engineering Institute of Canada (EIC), and a Nicola Copernicus Fellow of Italy.
He is the recipient of the esteemed Natural Sciences and Engineering Research Council (NSERC) of Canada E.W.R. Steacy Memorial Fellowship.
He is the only person in the world to hold both the IEEE Edwin Howard Armstrong Award and the IEEE Reginald Aubrey Fessenden Award,
named for the inventors of Frequency Modulation or FM, and Amplitude Modulation or AM, respectively.
Prof. Beaulieu is a Beijing University of Posts and Telecommunications BUPT Thousand-Talents Scholar.
He holds the third highest Web of Science ISI h-index in the world in the combined areas of communication theory and information theory.
Prof. Beaulieu was awarded the title ``State Especially Recruited Foreign Expert'' certified upon him by Minister of Human Resources and Social Insurance, and Vice Minister of the Organization Department, Yi Weimin.
Pro. Beaulieu is the recipient of the Royal Society of Canada Thomas W.  Eadie Medal, the M¨¦daille K.Y. Lo Medal of the EIC,
and was the subject of a TIME Magazine feature article.
He was also awarded the unique Special University Prize in Applied Science of the University of British Columbia,
and the J. Gordin Kaplan Award for Research of the University of Alberta.
\end{IEEEbiography}
\end{document}